\documentclass[conference]{IEEEtran}
\IEEEoverridecommandlockouts
\usepackage{cite}
\usepackage{amsmath,amssymb,amsfonts}
\usepackage{amsthm}
\usepackage{multirow}
\usepackage{url}
\usepackage{todonotes}

\usepackage{algorithmic}
\usepackage{graphicx}
\usepackage{textcomp}
\usepackage{hyperref}
\usepackage{xspace,tikz}
\usepackage{subcaption}
\usepackage[ruled,vlined,linesnumbered,noend]{algorithm2e}

\newtheorem{theorem}{Theorem}[section]
\newtheorem{corollary}{Corollary}[theorem]

\newtheorem{definition}{Definition}[section]
\usetikzlibrary{positioning}
\usepackage{xcolor}

\newcommand\bmaoged{\textsc{a$^*$-bmao}\xspace}
\newcommand\fori{\textsc{fori}\xspace}
\newcommand\forilp{\textsc{fori-lp}\xspace}
\newcommand\forithr{\textsc{fori-thr}\xspace}
\newcommand\bm{\textsc{bm}\xspace}
\newcommand\ls{\textsc{ls}\xspace}

\newcommand\forisim{\textsc{fori-sim}\xspace}
\newcommand\tlr{\textsc{tlr}\xspace}
\newcommand\aids{\textsc{aids}\xspace}
\newcommand\muta{\textsc{muta}\xspace}
\newcommand\protein{\textsc{prot}\xspace}

\usetikzlibrary{arrows.meta,positioning,calc}

\definecolor{lblA}{RGB}{51,102,204}   
\definecolor{lblB}{RGB}{240,140,40}   
\definecolor{ins}{RGB}{16,158,66}     
\definecolor{del}{RGB}{200,38,38}     
\definecolor{delred}{RGB}{206,44,44}    
\definecolor{delgrey}{RGB}{150,150,150}    
\definecolor{circred}{RGB}{206,44,44}   
\definecolor{circgrey}{RGB}{150,150,150} 

\tikzset{
  v/.style={circle,draw=black,inner sep=1.8pt,minimum size=17pt},
  va/.style={v,fill=lblA}, 
  vb/.style={v,fill=lblB}, 
  e/.style={line width=.9pt},
  map/.style={dotted,->,>=Latex,draw=black!55,line width=.9pt},
  insE/.style={delred,dashed,line width=1.4pt},
  delE/.style={delgrey, dashed, line width=1.2pt}, 
  delCircle/.style={circgrey, dashed, line width=1.2pt}, 
  insCircle/.style={circred, dashed, line width=1.2pt}, 
  insCirclelarge/.style={circred, dashed, line width=1.7pt}, 
  arr/.style={-Latex, line width=1.4pt}, 
  textlab/.style={font=\sffamily\Large}
}

\def\BibTeX{{\rm B\kern-.05em{\sc i\kern-.025em b}\kern-.08em
    T\kern-.1667em\lower.7ex\hbox{E}\kern-.125emX}}
\begin{document}

\title{Accelerating Graph Similarity Search through Integer Linear Programming
\thanks{This work was supported by the DFG under grant FOR-5361 – 459420781, by the BMBF (Germany) and state of NRW as part of the Lamarr-Institute, LAMARR22B, and partially funded by the European Union–NextGenerationEU under the Italian Ministry of University and Research (MUR) National Innovation Ecosystem grant ECS00000041–VITALITY–CUP E13C22001060006.  }
}

\author{\IEEEauthorblockN{Andrea D’Ascenzo}
\IEEEauthorblockA{\textit{Dept. of Computer Science} \\
\textit{Gran Sasso Science Institute}\\
L’Aquila, Italy \\
0000-0001-5612-0798}
\and
\IEEEauthorblockN{Julian Meffert}
\IEEEauthorblockA{\textit{Dept. of Computer Science} \\
\textit{University of Bonn}\\
Bonn, Germany \\
0009-0008-9670-4569}
\and
\IEEEauthorblockN{Petra Mutzel}
\IEEEauthorblockA{\textit{Dept. of Computer Science} \\
\textit{University of Bonn}\\
Bonn, Germany \\
0000-0001-7621-971X} 
\and
\IEEEauthorblockN{Fabrizio Rossi}
\IEEEauthorblockA{\textit{DISIM dept.} \\
\textit{University of L’Aquila, Italy}\\
L’Aquila, Italy \\
0000-0002-7495-390X}
}

\maketitle

\begin{abstract}
The {\em Graph Edit Distance} (GED) is an important metric for measuring the similarity between two (labeled) graphs.  It is defined as the minimum cost required to convert one graph into another through a series of (elementary) edit operations. Its effectiveness in assessing the similarity of large graphs is limited by the complexity of its exact calculation, which is NP-hard theoretically and computationally challenging in practice. The latter can be mitigated by switching to the {\em Graph Similarity Search} under GED constraints, which determines whether the edit distance between two graphs is below a given threshold. 

A popular framework for solving Graph Similarity Search under GED constraints in a graph database for a query graph is the {\em filter-and-verification} framework. Filtering discards unpromising graphs, while the verification step certifies the similarity between the filtered graphs and the query graph. To improve the filtering step, we define a lower bound based on an integer linear programming formulation. We prove that this lower bound dominates the effective branch match-based lower bound and can also be computed efficiently. Consequently, we propose a graph similarity search algorithm that uses a hierarchy of lower bound algorithms and solves a novel integer programming formulation that exploits the threshold parameter. An extensive computational experience on a well-assessed test bed shows that our approach significantly outperforms the state-of-the-art algorithm on most of the examined thresholds.
\end{abstract}

\begin{IEEEkeywords}
Graph similarity, graph edit distance, graph verification, integer linear programming
\end{IEEEkeywords}

\section{Introduction}

In graph-based data analysis, \textit{Graph Edit Distance} (GED) has emerged as a robust metric for quantifying the similarity between two labeled graphs.
It is defined as the minimum cost required to convert one graph into another through a series of edit operations - such as node and edge insertions, deletions and relabelling.
This approach allows to capture structural differences between the queried graphs, especially useful in scenarios where exact graph matching is hindered by noise or incomplete data, such as in pattern recognition, bioinformatics, computer vision, and graph databases, to name a few.

Despite its effectiveness, calculating the GED is an NP-hard problem~\cite{ZengTWFZ09}, posing significant challenges for scalability and efficiency when applied to large, complex graphs.
To mitigate these challenges, \textit{Graph Similarity Search under GED} shifts the focus to a more tractable problem: assessing whether the edit distance between two graphs falls below a given threshold.
This approach plays a pivotal role in enabling fast and accurate graph similarity searches, where the objective is to identify graphs within a large dataset that closely resemble a given query graph.

In particular, given a graph database $\mathcal{D}$, a query graph $q$ and a threshold parameter $\tau$, the aim of the graph similarity search problem is to select out of $\mathcal{D}$ a subset $\mathcal{S}$ of graphs whose GED to the query graph $q$ is no more than $\tau$.
The most effective framework to address the GED-based graph similarity search is the so-called \textit{filtering-and-verification}. This framework divides the computational burden in two phases: the filtering phase in which \textit{unpromising} graphs in $\mathcal{D}$ get discarded, and the verification phase in which the remaining graphs are further analysed to asses the GED with respect to the query graph $q$.

\begin{figure}
    \centering
    \includegraphics[width=0.7\linewidth]{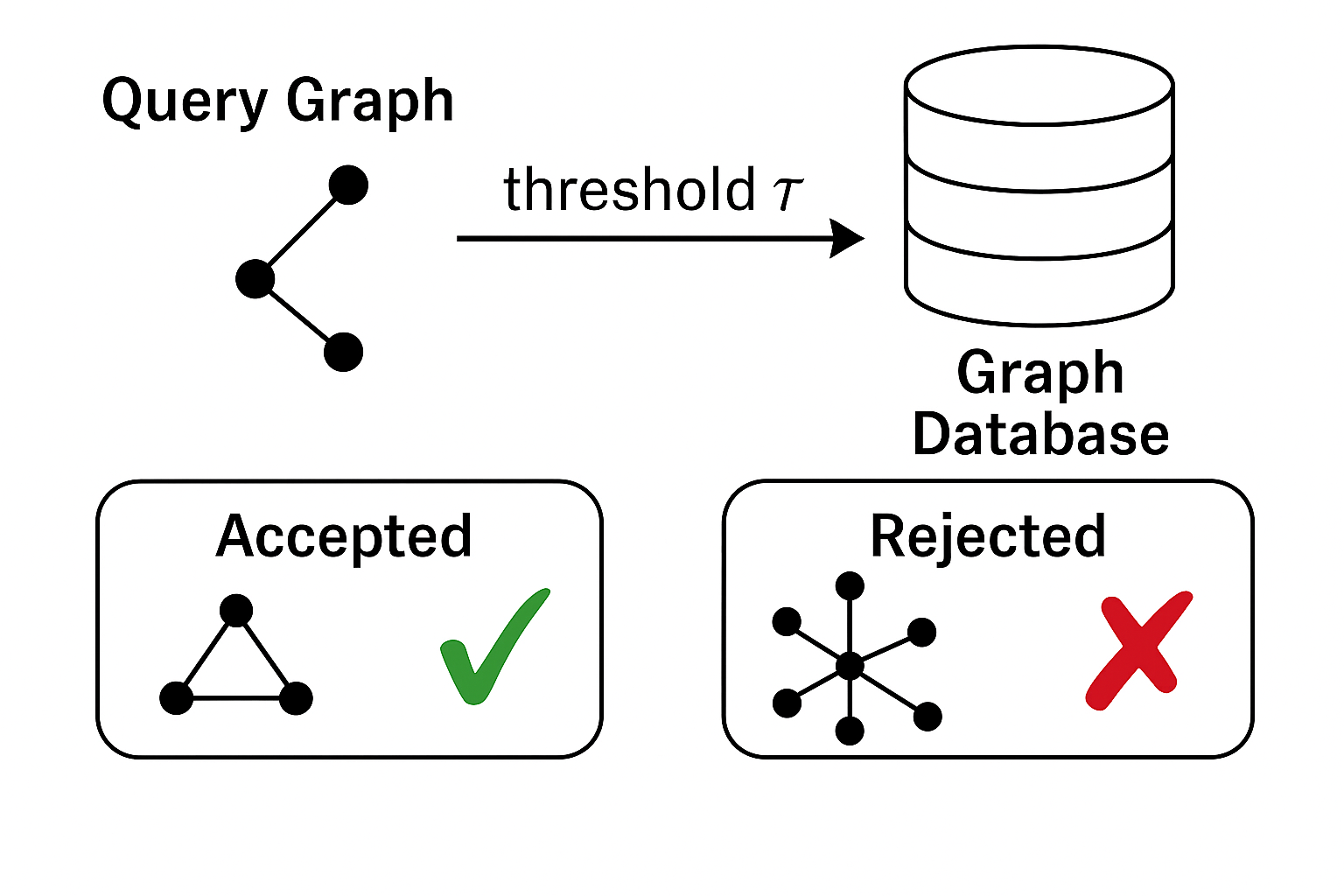}
    \caption{Graph similarity search.}
    \label{fig:verification-intro}
\end{figure}
Common approaches for the first phase are preprocessing techniques or fast GED lower bound computation. Tight and efficiently computable (global) lower bounds on the GED are essential for efficient graph similarity search since they allow to discard all graphs whose lower bound, w.r.t.\ the query graph, exceeds the threshold $\tau$. This greatly impacts the overall time needed to run a similarity search for each query graph, by reducing the number of graphs the exact GED has to be computed for, which is the most time intensive component of graph similarity search. On the other hand, approaches for the second phase mainly focus on speeding up $A^*$-based algorithms or providing stronger integer linear programming (ILP) formulations.
\subsection{Our contribution}
\begin{itemize}
\item We present an ILP-based approach for graph similarity search, that is applicable for general edit cost functions, based on the \fori ILP formulation proposed in~\cite{DAscenzoMMR25}.
\item Efficiency of the graph similarity search is strongly influenced by the quality of bounds. We prove that the linear relaxation of \fori provides a linear programming (LP)-based lower bound that dominates the branch match-based lower bound~\cite{zheng2014efficient}, which is recognized as one of the most effective, establishing a hierarchy on lower bound algorithms from the literature.
\item We provide a class of instances where the difference between the LP-based lower bounds and the branch match-based bounds gets arbitrarily large.
\item We propose an algorithm for graph similarity search that employs a hierarchy of lower bound algorithms before calling an exact GED algorithm based on an ILP formulation specifically tailored for the graph similarity search problem.
\item An extensive computational comparison on a well-assessed test bed demonstrates the practical impact of our theoretical results.
\end{itemize}

\section{Related work}

GED-based graph similarity search has received significant attention across diverse application domains such as cheminformatics, bioinformatics, pattern recognition, and computer vision \cite{liang2017similarity, ChangFYQZ23,ChangF0QZO20,blumenthal2017exact,Abu-AishehRRM15, ChenHHV19, KimC019, wang2012efficient,zhao2013efficient}. Existing works focus on the design of effective index structures, the design of lower bounds to increase the number of graphs that are discarded in the filtering phase, reducing the number of needed exact GED computations and proposing practical algorithms for GED verification or exact GED computation.
Proposed indexing structures include q-gram-based indexes~\cite{zhao2013efficient} - such as k-adjacent tree (k-AT)~\cite{wang2010efficiently}, subgraph-based index~\cite{liang2017similarity}, star structure-based index~\cite{wang2012efficient} and Pars~\cite{ZhaoXLZW18}. Index-based approaches have been shown to offer limited improvement compared to directly running, e.g., tree search-based algorithms~\cite{ChangFYQZ23, ChangF0QZO20}.
Several efficient algorithms for computing lower bounds on the GED have been proposed. These include but are not limited to label set-based lower bounds~\cite{blumenthal2017exact, ZengTWFZ09}, branch match-based lower bounds~\cite{zheng2014efficient, blumenthal2017exact} and linear programming-based lower bounds~\cite{Lerouge2017}. The state-of-the-art lower bound algorithms have been compared in an experimental survey~\cite{BlumenthalBGBB20}.

Despite the NP-hardness of computing the GED several practical algorithms have been developed. $A^*$-based-algorithms employ a best-first search strategy, usually differing by the heuristic used to compute the lower bound guiding the search. The \bmaoged algorithm~\cite{ChangFYQZ23} currently represents the state-of-the-art of tree search-based approaches both for graph similarity search and GED verification, outperforming depth-first approaches such as DF-GED~\cite{blumenthal2017exact, Abu-AishehRRM15}, CSI-GED~\cite{GoudaH16} and Beam-Stack-Search~\cite{ChenHHV19}. Recently Integer Linear Programming has been shown to outperform \bmaoged~\cite{DAscenzoMMR25}.

Recently, several machine learning approaches have emerged for estimating GED values~\cite{riba2021learning,PiaoXSRZC23}, or to be used as a heuristic for $A^*$-based algorithms~\cite{WangZYYY21}. While these methods efficiently approximate GED values, they are inherently heuristic and lack approximation guarantees. Moreover, most approaches predict only the GED value without providing an appropriate edit path, which is often crucial for interpretability and downstream applications.

\section{Preliminaries}\label{sec:preliminaries}
\subsection{Labeled (attributed) graphs and graph similarity}
We follow the notation used in~\cite{ChangFYQZ23}. 
Let $G = (V_G, E_G, l)$ be a simple, undirected, labeled\footnote{Throughout this paper, we use the term \textit{label} to refer to node or edge annotations, which are also commonly known as attributes in other domains.} graph, where $V_G$ is the set of nodes, $E_G$ is the set of edges, and $l: V_G \cup E_G \rightarrow \Sigma_V \cup \Sigma_E$ is a labeling function that assigns to each vertex $u\in V_G$ a label $l(u) \in \Sigma_V$ and to each edge $\{u,v\} \in E_G$ a label $l(u,v) \in \Sigma_E$.

A node $u \in V_G$ is considered a neighbor (or adjacent) to a node $v \in V_G$ if an edge $\{u,v\} \in E_G$ exists. For convenience, we may use the shorthand notation $uv$ in place of $\{u,v\}$. The neighborhood of a node $u$, denoted $\delta_G(u)$, is defined as the set of all nodes adjacent to $u$, i.e., $\delta_G(u) = \{v \in V_G \mid \{u,v\} \in E_G\}$, the set of edges incident to $u$ is denoted as $\Gamma_G(u) = \{\{u,v\} \in E_G \mid v \in V_G\}$.

To accommodate node and edge insertions or deletions, we extend the graph with a dummy node $\varepsilon_V$ and a dummy edge $\varepsilon_E$, resulting in the augmented sets $V_{G+\varepsilon} = V_G \cup \{\varepsilon_V\}$ and $E_{G+\varepsilon} = E_G \cup \{\varepsilon_E\}$. The specific labels and associated cost functions for nodes and edges depend on the dataset and application context.

Although most datasets in the literature consist of undirected graphs, the definitions and notation presented here can be readily adapted to directed graphs. In such cases, a directed edge from node $u$ to node $v$ is denoted by $(u,v)$, with $\delta_G^+(u)$ and $\delta_G^-(u)$ representing the sets of outgoing and incoming neighbors of $u$, respectively.

The \textit{graph edit distance} (GED) between two graphs $G$ and $H$, denoted as $GED(G,H)$, represents the minimum number of edit operations required to transform graph $G$ into graph $H$. These edit operations include: (1) inserting a labeled vertex; (2) deleting a labeled vertex; (3) changing the label of a vertex; (4) inserting a labeled edge; (5) deleting a labeled edge; and (6) modifying the label of an edge.
Each edit operation is associated with a non-negative \textit{edit cost}: $c_V: \Sigma _V \times \Sigma _V \to \mathbb{R}_{\geq 0}$ for node operations, $c_E: \Sigma _E \times \Sigma _E \to \mathbb{R}_{\geq 0}$ for edge operations.
Figure~\ref{fig:GEDexample} shows an example of the graph edit distance. In order to get graph $H$ from  graph $G$, five edit operations are needed: one node and one edge needs to be deleted from $G$ and two edges need to be inserted into $G$ in order to get $H$. Furthermore, a label change of the bottom node is needed. In the case of unit edit costs (i.e., all 1), the $GED(G,H)=5$.

\begin{figure}
\begin{center}
\scalebox{0.55}{  
\begin{tikzpicture}[x=1cm,y=1cm, line cap=round, line join=round]

\coordinate (L) at (0,0);

\coordinate (c1)  at ($(L)+(0,0)$);        
\coordinate (t1)  at ($(c1)+(0,1.8)$);     
\coordinate (r1)  at ($(c1)+(1.8,0)$);    
\coordinate (b1) at ($(c1)+(0,-1.8)$);  
\coordinate (l1) at ($(c1)+(-1.8,0)$); 
\node[va] (C1)  at (c1)  {};
\node[vb] (T1)  at (t1)  {};
\node[vb] (R1)  at (r1)  {};
\node[va] (B1) at (b1) {};
\node[va] (L1) at (l1) {};

\draw[e] (C1)--(T1);
\draw[e] (C1)--(R1);
\draw[e] (C1)--(B1);
\draw[e] (C1)--(L1);

\node[textlab,anchor=south] at ($(B1)-(0,1.5)$) {$G$};

\draw[arr] ($(c1)+(2.5,0)$) -- ($(c1)+(3.5,0)$);

\coordinate (M) at ($(c1)+(6.0,0)$);

\coordinate (c2)  at ($(M)+(0,0)$);
\coordinate (t2)  at ($(M)+(0,1.8)$);
\coordinate (r2)  at ($(M)+(1.8,0)$);
\coordinate (b2) at ($(M)+(0,-1.8)$);
\coordinate (l2) at ($(M)+(-1.8,0)$);

\node[va] (C2)  at (c2)  {};
\node[vb] (T2)  at (t2)  {};
\node[vb] (R2)  at (r2)  {};
\node[vb] (B2) at (b2) {};
\node[va] (L2) at (l2) {};

\draw[e] (C2)--(R2);
\draw[e] (C2)--(T2);
\draw[e] (C2)--(B2);

\draw[delE] (C2)--(L2);
\draw[insE] (T2)--(R2);
\draw[insE] (R2)--(B2);

\draw[delCircle] (L2) circle (0.36);
\draw[insCircle] (B2) circle (0.36);

\draw[arr] ($(M)+(2.5,0)$) -- ($(M)+(3.5,0)$);

\coordinate (R) at ($(M)+(5.5,0)$);

\coordinate (top3) at ($(R)+(0,1.4)$);
\coordinate (left3) at ($(R)+(-1.4,0)$);
\coordinate (right3) at ($(R)+(1.4,0)$);
\coordinate (bot3) at ($(R)+(0,-1.4)$);

\node[vb] (T3) at (top3) {};
\node[va] (L3) at (left3) {};
\node[va] (R3) at (right3) {};
\node[vb] (B3) at (bot3) {};

\draw[e] (L3)--(T3)--(R3)--(B3)--(L3);
\draw[e] (L3)--(R3);

\node[textlab,anchor=south] at ($(B3)-(0,1.5)$) {$H$};

\end{tikzpicture}
}
\caption{The Graph Edit Distance $GED(G,H)$ for $G$ and $H$ for unit edit costs is 5.} 
\label{fig:GEDexample}
\end{center}
\end{figure}
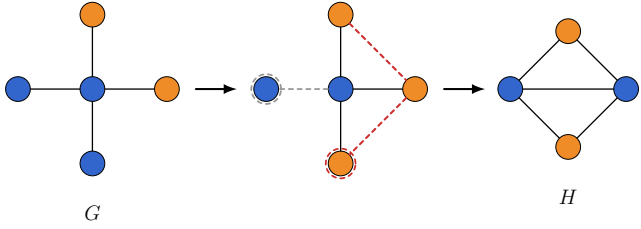

Having introduced the concept of GED, we can formally define the graph similarity search problem as follows.

\begin{definition}[Graph Similarity Search]
    Given a graph database $\mathcal{D}$, a query graph $Q$, and a threshold $\tau$, the graph similarity search asks to find all the graphs $H \in \mathcal{D}$ such that $GED(Q,H) \leq \tau$.
\end{definition}

Since the graph edit distance is employed as the distance metric, the graph similarity search problem studied in this paper is also known as graph edit similarity search~\cite{KimC019}.

\subsection{Integer Linear Programming}
An integer linear program ~\cite{conforti2014integer,wolsey2020integer} is defined as
\begin{equation}\label{eq:milp_problem}
    \min \{ c^T x \mid x \in \mathbb{Z}^n,\ Ax \geq b \}
\end{equation}
and comprises a linear objective function subject to a system of linear inequalities. Here, $A \in \mathbb{Q}^{m \times n}$ denotes the constraint matrix, $c \in \mathbb{Q}^n$ the cost vector, and $b \in \mathbb{Q}^m$ the right-hand side vector. A vector $\bar{x} \in \mathbb{Z}^n$ that satisfies all constraints $a_i^T \bar{x} \ge b_i$ for every $i = 1, \ldots, m$ is referred to as a \emph{feasible solution}. The collection of all feasible solutions to problem~\eqref{eq:milp_problem} is denoted by $\mathcal{X}$. An inequality of the form $\pi x \ge \pi_0$ is considered \textit{valid} for $\mathcal{X}$ if it holds for every $x \in \mathcal{X}$.
In the case $\mathcal{X} = \emptyset$, the ILP formulation is said to be infeasible.

The natural \emph{LP relaxation} of an ILP is obtained by replacing the integrality constraint $x \in \mathbb{Z}^n$ with the continuous constraint $x \in \mathbb{R}^n$. The resulting feasible region, defined as ${\mathcal P} = \{x \in \mathbb{R}^n \mid Ax \ge b\}$, forms a \emph{polyhedron} in $\mathbb{R}^n$, a bounded polyhedron is called a \textit{polytope}. For minimization problems such as GED, the optimal value of the LP relaxation $\nu(\mathcal P_M(I))$ of a formulation $M$ on instance $I$ serves as a lower bound to the optimal value of the corresponding integer program. 
Similarly, the optimal objective value of a linear objective function over a polyhedron $\cal P$ cannot be smaller compared to that over a polyhedron $\cal Q$ with $\cal P \subseteq \cal Q$.
In this work, we further restrict the variable domain to the binary set $\{0,1\}^n$.

An LP has a dual linear program. For the canonical LP (P):
$\min \{ c^T x \mid x \in \mathbb{Z}^n,\ Ax \geq b, x\ge 0 \}$ its dual is given by (D): 
$\max \{ y^T b \mid y^T A \leq c^T, y\ge 0 \}$. Every variable in (P) is associated with a dual constraint and vice versa. Furthermore, every constraint in (P) is associated with a dual variable and vice versa. The strong duality theorem states that if we have a primal feasible solution for a primal LP and a dual feasible solution of the corresponding dual LP, then the objective values of both solutions are the same if and only if both solutions are optimal solutions to (P) and (D).

\section{Methodology}

\subsection{Lower bounds from the literature}

In the following section we introduce the state of the art global lower bounds used in graph similarity search.

Let $L_V(X)$ and $L_E(Y)$ be the multi-set of vertex and resp. edge labels of $X\subseteq V_G$ and $Y\subseteq E_G$.
\begin{definition}[Label set-based lower bound~\cite{blumenthal2017exact}]
The label set-based lower bound (LS) of two graphs G, H is defined as
    $$\text{\ls}(G,H) = \Upsilon(L_V(V_G), L_V(V_H)) + \Upsilon(L_E(E_G), L_E(E_H))$$
where $\Upsilon(S_1,S_2) = \max\{|S_1|, |S_2|\} - |S_1 \cap S_2|$ denotes the edit distance between two multi-sets $S_1,S_2$.
\end{definition}
The \ls lower bound can be interpreted as counting the number of vertex and edge label mismatches between $G$ and $H$. \ls can be computed for graphs with categorical labels and unit edit costs in $O(n\log(n) + m\log(m))$, where $n=max\{|V_G|,|V_H|\}$ and $m=max\{|E_G|,|E_H|\}$ as sorted multi-sets can be intersected in $O(n)$~\cite{ZengTWFZ09} and has been employed in multiple graph similarity search algorithms~\cite{Abu-AishehRRM15, blumenthal2017exact,RiesenFB07}. 

The state-of-the-art algorithm for the graph similarity search with GED constraints \bmaoged uses an optimized version of the branch match-based lower bound, that calculates a matching between the vertices of $G,H$ based on the cost of matching the branch structures of their vertices. The \textit{branch} structure of a vertex $v$ is defined as $B(v) = (L_V(v), L_E(\delta(v))$ where $L_E(\delta(v))$ is the multi-set of labels of v's incident edges. We denote as $\Pi(G,H)$ the set of full mappings of the vertices $V_{G+\varepsilon}$ and $V_{H+\varepsilon}$ and respectively the set of full mappings between the edges of $\Gamma(i)\cup{\varepsilon}$ and $\Gamma(k) \cup \varepsilon$ as $\Pi(\Gamma(i),\Gamma(k))$.

\begin{definition}[Branch match-based lower bound~\cite{zheng2014efficient}]
The branch match-based lower bound \bm of two graphs $G,H$ is defined as 
{\small
    $$ \bm(G,H) = \min_{\pi \in \Pi(G,H)} \displaystyle\sum_{i\in V_G} c_{i,\pi(i)} + \min_{\sigma \in \Pi(\Gamma(i),\Gamma(k))} \sum_{ij \in \Gamma(i)} c_{ij,\sigma(ij)}.$$ 
    }

In the case of uniform edit costs it can be defined as 
{\small    $$ \text{BM}(G,H) = \min_{\pi \in \Pi(G,H)} \displaystyle\sum_{i\in V_G} \delta_{l(i)\neq l(k)} + \frac{1}{2}\Upsilon(L_E(\Gamma(i),L_E(\Gamma(\pi(i))) $$
}
where $\Upsilon(\cdot,\cdot)$ is multiplied with a factor of $\frac{1}{2}$ because the edge $(i,j)\in E_G$ can produce costs both in the matching of $B(i)$ and $B(j)$ and $\delta_{cond}$ is the Dirac delta $\delta_{cond} = 1$ if $cond$ evaluates as true and $\delta_{cond} = 0$ otherwise.
\end{definition}
The \bm lower bound can be computed for general edit costs in $O((|V_G|+|V_H))^3 + |V_G||V_H|\Delta_{\min}^2\Delta_{\max})$~\cite{BlumenthalBGBB20} where 
$$
\begin{cases}
\Delta_{\min} = \min \{ \displaystyle \min_{v \in V_G} |\delta(v)|, \min_{w \in V_H} |\delta(w)|\}\\
\Delta_{\max} = \max \{ \displaystyle \max_{v \in V_G} |\delta(v)|, \max_{w \in V_H} |\delta(w)|\}
\end{cases}
$$

\noindent
For unit edit costs it can be computed in $O(n\log n + m\log m)$, $n=\max\{|V_G|,|V_H|\}$, $m=\max\{|E_G|,|E_H|\}$~\cite{zheng2014efficient}.

Any ILP formulation for the graph edit distance lends itself to obtain a lower bound on the GED by relaxing the integrality constraint on the variables, which has already been suggested by \cite{BlumenthalBGBB20, Lerouge2017}. The optimal solution to a linear programming model can be computed efficiently e.g., using interior point methods in runtime $O(n^3L)$~\cite{Renegar88}, using $L$-bit numbers and $n$ variables, although in practice, e.g., the simplex algorithm is more efficient than in the theoretical worst case~\cite{SpielmanT04}, with an average case runtime of $O(n\log n))$~\cite{Borgwardt82}. In this work, we will use the \fori formulation as the basis of our LP-based approaches.

\subsection{The \fori formulation}
Integer Linear Programming formulations have emerged as an effective and practical approach for solving the GED problem. The \fori formulation has been proven to be the strongest formulation of the known ones.~\cite{DAscenzoMMR25} 

The core idea of the formulation is to orient the edges of the input graphs. In $G$, every undirected edge $\{i,j\}$ is oriented so that $i<j$, and the resulting directed graph is denoted by $\overrightarrow{G}$.
Moreover, in $H$, two arcs $(k,l)$ and $(l,k)$ are introduced for each edge $\{k,l\}\in E_H$ with $k\ne l$, leading to the graph $\overleftrightarrow{H}$.

To model node assignments, a binary variable $x_{i,k}$ is defined for each pair $i \in V_G$, $k \in V_H$, where $x_{i,k} = 1$ indicates that node $i$ is mapped to node $k$, and $x_{i,k} = 0 $ otherwise. 
To model edge assignments, $z$ variables $z_{ij,kl}$ are introduced for the set of all arcs $(i,j)\in E_{\overrightarrow{G}}$ ($i<j$) and all arcs $(k,l)\in E_{\overleftrightarrow{H}}$. A variable $z_{ij,kl}$ is set to 1 if arc $(i,j) \in E_{\overrightarrow{G}}$ gets mapped to arc $(k,l) \in E_{\overleftrightarrow{H}}$ and 0 otherwise.
Finally, the \fori cost function inherently accounts for node and edge deletions and insertions, following the approach in~\cite{Lerouge2017}. The constant term
\[
K = \sum_{i \in V_G} c_{i,\varepsilon} + \sum_{k \in V_H} c_{\varepsilon,k} + \sum_{ij \in E_G} c_{ij,\varepsilon} + \sum_{kl \in E_H} c_{\varepsilon,kl}
\]
represents the cumulative cost of removing and subsequently adding every node and edge, and is incorporated into the objective function. Mapping costs are defined as $\bar{c}_{i,k} = c_{i,k} - c_{i,\varepsilon} - c_{\varepsilon,k}$ for all node pairs $(i,k) \in V_G \times V_H$, and $\bar{c}_{ij,kl} = c_{ij,kl} - c_{ij,\varepsilon} - c_{\varepsilon,kl}$ for all edge pairs $(ij,kl) \in E_G \times E_H$. Consequently, when $x_{i,k} = 1$, the objective function includes the cost of mapping $i$ to $k$ while offsetting the deletion cost of $i \in V_G$ and the insertion cost of $k \in V_H$; the same logic applies to $y_{ij,kl} = 1$ for edge mappings.
The \fori formulation is reported in Figure~\ref{fig:modelFORI}. Constraints~(\ref{constr:fori_first_assignment}) and (\ref{constr:fori_second_assignment})  model the assignment between the nodes, and make sure that a node is not mapped to (resp.\ from) more than one node. 
The constraints~(\ref{constr:fori_first_impl}), (\ref{constr:fori_second_impl}), and (\ref{constr:fori_third_impl}) link the edge mapping variables with the node mapping variables and make sure that edges can be mapped only if the corresponding end nodes have been mapped. Let \forilp denote the natural LP relaxation of \fori obtained by replacing $x,z \in \{0,1\}$ by $x,z\in [0,1]$.

\begin{figure}[t]
\begin{subequations}
\begin{align}
\noalign{\noindent\(\quad\min\displaystyle\sum_{i\in V_G}\sum_{k\in V_H}\bar c_{i,k} x_{i,k}
       + \sum_{(i,j)\in E_{\overrightarrow{G}}}\sum_{(k,l)\in E_{\overleftrightarrow{H}}}\bar c_{ij,kl} z_{ij,kl} + K\)}\notag \\
\text{s.t.} \quad 
    & \sum_{k \in V_H} x_{i,k} \leq 1 \quad \forall \ i \in V_G \label{constr:fori_first_assignment} \\
    & \sum_{i \in V_G} x_{i,k} \leq 1 \quad \forall \ k \in V_H \label{constr:fori_second_assignment} \\
    & \sum_{l \in \delta^+_{\overleftrightarrow{H}}(k)} z_{ij,kl} \leq x_{i,k} \quad \forall \ k \in V_H, \ (i,j) \in E_{\overrightarrow{G}} \label{constr:fori_first_impl} \\
    & \sum_{l \in \delta^-_{{\overleftrightarrow{H}}}(k)} z_{ij,lk} \leq x_{j,k} \quad \forall \ k \in V_H, \ (i,j) \in E_{\overrightarrow{G}} \label{constr:fori_second_impl} \\
    & \sum_{j \in \delta^+_{\overrightarrow{G}}(i)} z_{ij,kl} 
      + \sum _{j \in \delta^-_{\overrightarrow{G}}(i)} z_{ji,lk} \leq x_{i,k} \notag\\
    & \qquad \qquad\qquad\qquad \forall \ i \in V_G, \ (k,l) \in E_{\overleftrightarrow{H}} \label{constr:fori_third_impl} \\
    & x \in \{0,1\}^{|V_G|\cdot |V_H|} \label{fori:x_integrality}\\
    & z \in \{0,1\}^{|E_G|\cdot 2|E_H|} \label{fori:z_integrality}
\end{align}
\end{subequations}
\caption{Formulation \fori~\cite{DAscenzoMMR25}.}
\label{fig:modelFORI}
\end{figure}

Relaxing the integrality constraints in ILP formulations enables the efficient computation of a lower bound on the GED~\cite{BlumenthalBGBB20}. 
Within the filtering-and-verification framework, a strong lower bound can significantly improve efficiency by increasing the number of filtered graphs - if the bound exceeds the predefined threshold, the candidate graph can be discarded without further evaluation - and speeding up the GED verification, e.g., by using it as the heuristic in an $A^*$-based algorithm.

\subsection{Theoretical justification: Strong FORI-based lower bounds}
We begin by motivating the adoption of the \fori formulation within the graph similarity search framework through a first-time analysis of its linear relaxation. In particular, we compare the quality of the lower bound provided by \forilp with that of the state-of-the-art methods. We prove that \forilp dominates the branch match-based lower bound for all instances.

Moreover, we provide an instance (two classes of graphs on $n$ vertices each) for which \forilp provides the optimum solution value of $2n-5$, and the \bm lower bound $n-2$, hence with increasing $n$ the difference gets arbitrarily large.

\begin{theorem}
\label{thm:lb}
Let $G=(V_G,E_G)$ and $H=(V_H,E_H)$ arbitrary labeled graphs together with the unit cost function, then it holds that the lower bound value of \bm is not larger than that of \forilp.
\end{theorem}
\begin{proof}
    To prove the theorem, we construct from the solution calculated by \bm a solution to a relaxation of the (F1) ILP formulation~\cite{Lerouge2017}, see Figure~\ref{fig:F1}. Let $\mathcal{P}_{F1}$ denote its corresponding LP relaxation polytope. We will see that $\mathcal{P}_{F1}$ is contained in the polytope corresponding to the solutions of the \bm heuristic, denoted as $\mathcal{P}_{\bm}$, i.e. $\mathcal{P}_{F1} \subseteq \mathcal{P}_{\bm}$ and follow that, because \fori's LP relaxation polytope $\mathcal{P}_{\fori}$ is strictly contained in $\mathcal{P}_{F1}$ (Lemmas 6.4 + 6.5~\cite{DAscenzoMMR25}), that  $\bm(G,H) \leq \nu(\mathcal{P}_{\fori}(G,H))$.
    
    Let $x^*$ be the solution calculated by \bm, defined as $x^*_{i,k} = 1$ if node $i$ is matched to node $k$ and $x^*_{i,k} = 0$ otherwise. We now use $x^*$ to formulate edge variables $y^*_{ij,kl}$ that capture the cost of matching the branch structures $B(i)$ and $B(k)$. Therefore, let $\sigma \in \Pi(\Gamma(i),\Gamma(k))$ be an optimum solution to the bipartite edge matching of the branch structures. We introduce variables $y^*_{ij,kl} \in [0,1]$ for each possible edge mapping in $E_{G+\varepsilon} \times E_{H+\varepsilon}$, setting $y^*_{ij,\sigma(ij)} = 0.5$ if $x^*_{i,k} + x^*_{j,k} + x^*_{i,l} + x^*_{j,l} = 1$, $y^*_{ij,\sigma(ij)} = 1$ if the sum is equal to $2$ and $y^*_{ij,kl} = 0$ otherwise. 
    
    We turn now to the ILP formulation (F1). Variables $x_{i,k}, y_{ij,kl}$ capture the node and edge assignments similarly to \fori, constraints (\ref{F1:eq1})-(\ref{F1:eq4}) ensure that every node and edge is either mapped or deleted. The topological constraints (\ref{F1:ineq5})-(\ref{F1:ineq6}) ensure that an edge can only be mapped if both of its endpoints are mapped. This means that in order for the \bm solution $(x^*, y^*)$ to be feasible for (F1) we need to relax the topological constraints, as e.g., $y_{ij,kl} = 0.5$, $ x_{i,k} = 1$, $x_{j,k} = x_{i,l} = x_{j,l} = 0$ would violate constraint (\ref{F1:ineq6}). We replace the topological constraints with the inequalities.
    \begin{equation}
        y_{ij,kl} \leq \frac{1}{2}(x_{i,k} + x_{j,k} + x_{i,l} + x_{j,l})   \label{F1:relax}
    \end{equation}
    As inequality (\ref{F1:relax}) is obtained by summing up the pair of topological constraints for each $(ij,kl) \in E_{G} \times E_{H}$ and dividing both sides by 2, the resulting inequality cannot be stronger than the original constraints. 
    Since $x^*$ corresponds to node matching and $y^*$ to an edge matching, the constraints (\ref{F1:eq1})-(\ref{F1:eq4}) are satisfied, by the above definition the variables $y^*_{ij,kl}$ satisfy the constraint (\ref{F1:relax}) and thus $(x^*, y^*)$ is a solution to the system of (in)equalities (\ref{F1:eq1})-(\ref{F1:eq4}) together with (\ref{F1:relax}). Defining the cost to be $d_{i,k} = \delta_{l(i)\neq l(k)}$ and $d_{ij,kl} = \delta_{l(i,j) \neq l(k,l)}$ we obtain that the objective function value of $(x^*, y^*)$ is equal to $\bm(G,H)$, as the objective function coefficients count the number of vertex label mismatches and each edge label mismatch scaled by a factor of $\frac{1}{2}$ via the right hand side of inequality (\ref{F1:relax}). We denote the polytope corresponding to the relaxed system of inequalities by $\mathcal{P}_{\bm}$. Because we replaced the topological constraints (\ref{F1:ineq5})-(\ref{F1:ineq6}) with their linear combination (\ref{F1:relax}), we obtain that $\mathcal{P}_{F1} \subseteq \mathcal{P}_{\bm}$.
    This also means that $\bm(G,H) \leq \nu(\mathcal{P}_{F1}(G,H))$ and thus Lemma 6.5 in \cite{DAscenzoMMR25} yields that $\bm(G,H) \leq \nu(\mathcal{P}_{\fori}(G,H))$. 
    \end{proof}
    
\begin{figure}
    \begin{subequations}
    \begin{align}
    \noalign{\noindent\(\quad\min\displaystyle\sum_{i\in V_{G+\varepsilon}}\sum_{k\in V_{H+\varepsilon}}d_{i,k}x_{i,k} + \sum_{ij\in E_{+\varepsilon}}\sum_{kl\in E_{H+\varepsilon}}d_{ij,kl}y_{ij,kl}\)}\notag\\[-2mm]
        \text{s.t.\quad}\sum_{k \in V_{H+\varepsilon}}x_{i,k}&=1 \quad \forall \ i\in V_G\label{F1:eq1}\\
    \quad \sum_{i \in V_{G+\varepsilon}}x_{i,k}&=1 \quad \forall \ k\in V_H\label{F1:eq2}\\
    \quad \sum_{kl\in E_{H+\varepsilon}}y_{ij,kl}&=1 \quad \forall \ ij\in E_G\label{F1:eq3}\\
    \quad \sum_{ij\in E_{G+\varepsilon}}y_{ij,kl}&=1 \quad \forall \ kl\in E_H\label{F1:eq4}\\
    \quad y_{ij,kl} \leq x_{i,k} &+ x_{j,k} \quad \forall \ ij \in E_G, \ kl \in E_H\label{F1:ineq5}\\
    \quad y_{ij,kl} \leq x_{i,l} &+ x_{j,l} \text{\space}\quad \forall \ ij \in E_G, \ kl \in E_H\label{F1:ineq6}\\
    x&\in \{0,1\}^{|V_G||V_H|+|V_G|+|V_H|}\\
    y&\in \{0,1\}^{|E_G||E_H|+|E_G|+|E_H|}
    \end{align}
    \end{subequations}
    \caption{ILP Formulation (F1)~\cite{Lerouge2017}.}
    \label{fig:F1}
\end{figure}

Note, that \forilp can be straightforwardly extended to an anchor-aware lower bound, taking partial mappings into account by enforcing $x_{i,k}=1$ based on the already fixed vertices.  

We proceed by providing a class of GED instances for which the difference between the lower bound provided by \forilp and that given by \bm  gets arbitrarily large, which is illustrated in Figure~\ref{fig:thm_lb}.

\begin{theorem}
\label{thm:lb_large}
For the instance $G=(V_G,E_G)$ an unlabeled star ($\mathcal{S}_n$) on $n \ge 3$ vertices and $H=(V_H,E_H)$ an unlabeled cycle ($\mathcal{C}_n$) of $n$ vertices with unit edit costs, the optimal value of $\forilp$ is equal to $2n-5$ and that of \bm is $n-2$.
Thus, for any given number $R$, there exists an instance ($n$ large enough) for which the difference of the two bounds gets larger than $R$.
For $n\rightarrow \infty$, the lower bound provided by \forilp is 2 times larger than that of \bm.
\end{theorem}

\begin{proof} 
In order to show that the optimal value of \forilp is $2n-5$, we will present a feasible LP solution $\bar x$ attaining this value. According to the strong duality theorem, $\bar x$ is an optimal solution of the LP, if and only if there exists a feasible solution to the dual LP, and both solutions have the same objective function value.

Let $1$ be the index of the center vertex of the star graph. We are going to build the primal solution so that node $i$ of the star $\mathcal{S}_n$ is mapped to node $i$ of the cycle $\mathcal{C}_n$, and the two edges $\{1,2\}$ and $\{1,n\}$ of the star are mapped to the edges $\{1,2\}$ and $\{1,n\}$ of the cycle. To do so, we fix $x_{i,i} = 1$ for all $i=1,\dots,n$, and $x_{i,j} = 0$ for all $i\neq j, \ i,j = 1,\dots, n$. Finally, we set $z_{12,12}=z_{1n, 1n}=1$, and all other $z_{ij,kl}$ variables to 0.
Since on this instance the constant $K$ in the objective function of \forilp is equal to $4n-1$, and the edit cost of matching $n$ (unlabeled) nodes and two (unlabeled) edges is $-2$ for each operation, we get an objective value of $(4n-1)-2(n+2)=2n-5$, where $-2(n+2)$ comes from node and edge mappings. Checking the constraints we see that this solution is feasible for \forilp. 

In order to build the dual LP, we bring \forilp into the canonical LP form $\min c^T, Ax\ge b$ by multiplying the constraints with $(-1)$. The resulting dual LP (DF) is shown in Figure~\ref{fig:dual-FORI-LP}. We have dual variables $u\in Q^{|V_G|}$ and $v\in Q^{|V_H|}$ arising from the primal (node mapping) constraints~(\ref{constr:fori_first_assignment}) and (\ref{constr:fori_second_assignment}).
For the three topological constraints~(\ref{constr:fori_first_impl}), (\ref{constr:fori_second_impl}), and (\ref{constr:fori_third_impl}) of \forilp we have the dual variables $r\in Q^{|V_H||E_G|}$, $s\in Q^{|V_H||E_G|}$, and $t\in Q^{|V_G|2|E_H|}$. Consider the following dual solution:  $u_1=6, u_2,\ldots,u_n=2$ and all $v_i=0$ for all $i=1,\ldots,n$. Also all the $r$ and $s$ variables and almost all the $t$ variables are set to 0 with the exception of $t_{1,kl}$, which is set to 2 for all $(k,l) \in E_{\overleftrightarrow{H}}$. Since the right hand side of the constraints~(\ref{constr:dualfori-1}) and (\ref{constr:dualfori-1}) is $(-2)$ (corresponding to the costs of the primal variables), we can observe that all the constraints in the dual LP (see Fig.\ \ref{fig:dual-FORI-LP}) are satisfied. E.g., for $i=1$, constraint~(\ref{constr:dualfori-1}) evaluates to $-6+2+2=-2$, since node $k$ has exactly two outgoing arcs in $E_{\overleftrightarrow{H}}$. This follows from the fact that each undirected edge in $H$ corresponds to two directed copies in the \forilp model, and hence also in its dual. For $i\not=1$, all $r,s$ and $t$ variables are 0, which implies that constraint~(\ref{constr:dualfori-1}) takes the value -2.
It remains to show that the objective value is equal to that of the primal solution $\bar{x}$.
We have $K -\sum_{i \in V_G} v_i = K -6 -2(n-1)=-4 - 2n$ with $K=4n -1$, hence we get an 
objective value of $2n -5$, which by strong duality yields that this is the optimal solution value of \forilp.

Turning to the value of \bm, we observe that every possible node mapping produces the same cost, as for $i=1$ there must be $n-2$ edges that are deleted, which produces a cost of $\frac{1}{2}$ for each deleted edge for a total of $\frac{n-3}{2}$, and for all nodes in the star except for the center node one edge needs to be inserted which produces a cost of $\frac{n-1}{2}$ in total. This leads to a total cost of $\frac{n-3}{2} + \frac{n-1}{2} = n-2$, which yields the result.
\end{proof}

\begin{figure}
\begin{subequations}
\begin{align}
\noalign{ \[ \text{(DF)} \qquad K + \quad \max  \big(-\sum_{i \in V_G} u_i  - \sum_{k \in V_H} v_k \big) \] } \notag \\[-2mm]
\noalign{ \[ \text{s.t.} \quad -u_i - v_k +\sum_{j \in \delta^+_G(i)} r_{ij,k}+ \sum_{j \in \delta^-_G(i)} s_{ji,k} \] }\notag \\[-4mm]
  + \sum_{l \in \delta^+_H(k)} t_{i,kl} &\le \bar{c}_{i,k}  \qquad \forall i \in V_G, k \in V_H \label{constr:dualfori-1}  \\
 -r_{ij,k} - s_{ij,l} - t_{i,kl} &- t_{j,kl} \le \bar{c}_{ij,kl}  \qquad \notag\\
\qquad\qquad\qquad\quad &\forall (i,j) \in E_{\overrightarrow{G}}, (k,l) \in E_{\overleftrightarrow{H}}  \label{constr:dualfori-2}\\
  \qquad u, v, r, s, t &\ge 0  \label{constr:dualfori-3}
\end{align}
    \end{subequations}
    \caption{The dual of the FORI LP-relaxation.}
    \label{fig:dual-FORI-LP}
\end{figure}

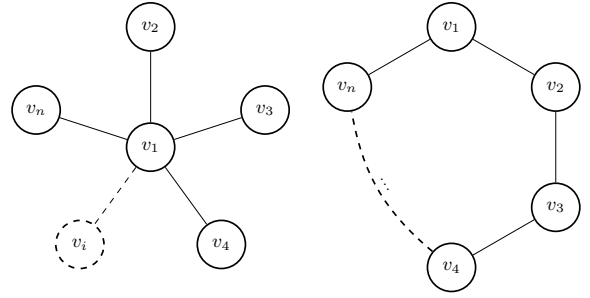
\begin{figure}
\begin{center}
\scalebox{0.8}{  
\begin{tikzpicture}[x=2cm,y=2cm,
  node/.style={circle,draw=black,thick,minimum size=8mm,inner sep=0pt,font=\small},
    nodedashed/.style={circle,draw=black,style=dashed,thick,minimum size=8mm,inner sep=0pt,font=\small},
  every edge/.style={thick}
  ]

\coordinate (CL) at (0,0);
\node[node] (L1) at (CL) {$v_1$};

\foreach \i [count=\j from 2] in {90, 18, 306}{
  \node[node] (L\j) at ({\i}:1) {$v_{\j}$};
  \draw (L1) -- (L\j);
}
 \node[node] (L6) at ({162}:1) {$v_n$};
  \draw (L1) -- (L6);

\node[nodedashed] (L5) at ({234}:1) {$v_i$};
\draw[dashed] (L1) -- (L5);

\begin{scope}[xshift=5cm]
  \coordinate (v1) at (90:1);
  \coordinate (v2) at (30:1);
  \coordinate (v3) at (-30:1);
  \coordinate (v4) at (-90:1);
  \coordinate (vn) at (150:1); 

  \node[node] (N1) at (v1) {$v_1$};
  \node[node] (N2) at (v2) {$v_2$};
  \node[node] (N3) at (v3) {$v_3$};
  \node[node] (N4) at (v4) {$v_4$};
  \node[node] (Nn) at (vn) {$v_n$};

  \draw (N1)--(N2)--(N3)--(N4);
  \draw[dashed,thick,bend left=20] (N4) to node[font=\scriptsize,above,sloped,pos=.5]{\dots} (Nn);
  \draw (Nn)--(N1);
\end{scope}

\end{tikzpicture}
 }
  \caption{Star graph $S_n$ and cycle graph $C_n$ used in Theorem~\ref{thm:lb}.}
  \label{fig:thm_lb}
\end{center}
\end{figure}

From Theorem~\ref{thm:lb} and Lemma 4.3 in~\cite{ChangFYQZ23} we can derive a hierarchy on the branch match-based lower bound \bm, the label set-based lower bound \ls, and the lower bound provided by the linear relaxation of \fori.
\begin{corollary} The following hierarchy on the lower bounds provided by each method holds for all instances:
    $$\forilp \geq \bm \geq \ls.  $$   
\end{corollary}

 The discussed lower bounds achieve different trade offs between the tightness of the bound and efficiency, prompting the idea to employ a hierarchy of multiple lower bounds with increasing tightness but decreasing efficiency in the filtering phase, exploiting that with small $\tau$ many graphs might successfully get filtered by the simplest and fastest of heuristics, before needing to compute tighter, less efficient lower bounds or even the exact GED.

\subsection{FORI-based GED Verification}
If an input graph passes the filtering stage - meaning that its GED lower bound with respect to the query graph is below the threshold - it proceeds to the verification phase.
In this phase, the exact GED should be computed to determine whether the actual distance lies below or above the threshold.
In what follows we adapt the \fori formulation, that was shown in~\cite{DAscenzoMMR25} to outperform the state-of-the-art method \bmaoged in computing the exact GED, to address the GED verification task.

Specifically, we introduce a threshold constraint that enforces the objective function to remain below the threshold $\tau$.
This modification enables the ILP solver to terminate once it finds any feasible solution or stop the computation once it is determined that the model is infeasible.
This approach is advantageous because it eliminates the need to compute the exact GED value; it suffices to identify any solution with objective value below the threshold.

Constraint (\ref{constr:fori_flat}) is our novel threshold constraint that limits the objective function of \fori, the optimal value of which is equal to the GED, to stay under the specified threshold $\tau$, lest otherwise the model is determined infeasible.

\begin{equation}
    \sum_{i\in V_G}\sum_{k\in V_H}\bar c_{i,k}\cdot x_{i,k}
      + \sum_{(i,j)\in E_{\overrightarrow{G}}}\sum_{(k,l)\in E_{\overleftrightarrow{H}}}\bar c_{ij,kl}\cdot z_{ij,kl} + K \leq \tau. \label{constr:fori_flat}
\end{equation}

We denote the model obtained by extending \fori with constraint~(\ref{constr:fori_flat}) as \forithr.
Algorithm~\ref{algo:foriflat}, referred to as \forisim, outlines our approach to the graph similarity problem. Given a query graph $Q$, a graph dataset $\mathcal{D}$, and a similarity threshold $\tau$, the algorithm iterates over each graph $H \in \mathcal{D}$ and initially computes $\ls(G,H)$, of the lower bound is below the threshold $\tau$ it repeats the process with the $\bm$ lower bound and finally \forilp. If the lower bound provided by any of these algorithms is larger than $\tau$ the graph $H$ can be safely discarded.
Furthermore, the inclusion of the new inequality ensures that once a feasible solution of \forithr is identified, the graph $H$ can be accepted and added to the set $\mathcal{A}$ of accepted graphs (see Line~\ref{code:accepted} in Algorithm~\ref{algo:foriflat}).
Conversely, if the ILP is found to be infeasible, the graph $H$ is excluded from further consideration.

The next section presents a comparative analysis between \forisim and the state-of-the-art algorithms for graph similarity search.

\begin{algorithm}[t]
    \SetAlgoLined\small
	\KwIn{Query graph $Q=(V_Q,E_Q, L_V, L_E)$, graph dataset $\mathcal{D}$, similarity threshold $\tau$.}
	\KwOut{Set $\mathcal{A} \subseteq \mathcal{D}$ of graph $H$ having $GED(Q,H) \leq \tau$.}
	\ForEach{$H \in \mathcal{D}$}{
        \ForEach{\textsc{alg} $\in \{\ls, \bm, \forilp\}$}{
            $lb \gets $ \textsc{alg}$(Q,H)$\;
            \If{$lb > \tau$}{
                \textbf{goto} \ref{code:exit}\;
            }
        }
        Run $\forithr(Q,H,\tau)$\;
        \If{$\forithr(Q,H,\tau)$ is feasible}{
            $\mathcal{A} \gets \mathcal{A} \ \cup \ H $\label{code:accepted}\;
        }
        \Else{
        Discard $H$\label{code:exit}\;
        }
        }
    \caption{Algorithm \forisim.} 
    \label{algo:foriflat}
\end{algorithm}%

\section{Experimental analysis}
We evaluate the presented approaches with respect to the following research questions:
\begin{itemize}
    \item \textbf{Q1}: Will \forilp yield better lower bounds than the state-of-the-art algorithms on real-world graph topologies?
    \item  \textbf{Q2}: How does our algorithm \forisim compare to \bmaoged in graph similarity search on unit edit costs?
    \item \textbf{Q3}: How efficient is our algorithm \forisim in the graph similarity search on non-uniform edit cost functions?
\end{itemize}

\noindent\textbf{Datasets.} We perform experiments on three datasets: \aids, Mutagenicity (referred as \muta from now on), and Protein (shortened as \protein from now on), collected in the IAM Graph Database repository~\cite{RiesenB08}, which include graphs with both node and edge labels; 
the \aids and \muta datasets represent molecular structures, where nodes are labeled with one of 13 chemical symbols and edges indicate valence values of 1, 2, or 3.
The graphs in the \protein dataset represent proteins annotated with their corresponding EC classes~\cite{SchomburgCEGHHS04}. Each node is labeled by a tuple $(t, s)$, where $t$ denotes the structural type (helix, sheet, or loop), and $s$ encodes the amino acid sequence. The dataset includes 8204 distinct protein sequences. Edges between nodes capture structural and/or sequential relationships and are labeled with tuples $(t_1, t_2)$, where $t_1$ and $t_2$ specify the types of the first and second connections between nodes $u_i$ and $u_j$, respectively; $t_2$ may be \texttt{null}. In total, five unique edge types are observed across the dataset.

Graph files are obtained from the \textit{GEDLIB} library by Blumenthal et al.~\cite{BlumenthalBGBB20}.
Table~\ref{tab:dataset_stats} summarizes datasets statistics. $|\mathcal{D}|$ gives the number of graphs in the dataset. $|V|$ and $|E|$ denote the number of nodes and edges per graph, respectively. $|\Sigma_V|$ and $|\Sigma_E|$ represent the sizes of the node and edge label alphabets.

\begin{table}[h]
\caption{Dataset statistics.}
\label{tab:dataset_stats}
\centering
\footnotesize
\setlength{\tabcolsep}{4pt}
\renewcommand{\arraystretch}{1.1}
\begin{tabular}{l|r|r|r|r|r|r|r}
\textbf{Dataset} & $|\mathcal{D}|$ & avg $|V|$ & max $|V|$ & avg $|E|$ & max $|E|$ & $|\Sigma_V|$ & $|\Sigma_E|$ \\
\hline
\aids          & 2000  & 15  & 95   & 16  & 103  & 38 & 3 \\
\muta  & 4339  & 30  & 417  & 30  & 112  & 14 & 3 \\
\protein  & 600  & 32  & 126  & 62  & 149  & 8249 & 17489 \\
\end{tabular}
\end{table}

From each dataset, we selected 10 graphs to serve as query graphs. Their average sizes are as follows: for the \aids dataset, query graphs contain an average of 36 nodes and 38 edges; for \muta, the averages are 29 nodes and 30 edges; for \protein, query graphs have on average 34 nodes and 64 edges.

\noindent\textbf{Computational setting.} The experiments were run on a MacBook M4 Pro with a 12-cores CPU with 48Gb of RAM, macOS Sequoia 15.6.  
We use the implementation of the \bm lower bound from \textit{gedlib}~\cite{BlumenthalBGB19}, we adapted the implementation of \ls provided in~\cite{ChangFYQZ23} and use without change the authors implementation of \bmaoged from the same source. All algorithms are implemented in C++ and compiled using Apple clang version 17, with flag -O3.
We use Gurobi 12.0.3~\cite{gurobi} to solve the ILPs.
Access to our implementation and datasets is provided via the following link \url{https://github.com/D-hash/FORI-SIM}.

\subsection*{\textbf{Answering Q1:} Lower bound comparison}
In this section, we compare the performance of \forilp with state-of-the-art global lower bounds for $GED(G,H)$. We restrict our evaluation here to \muta and \aids using unit edit cost, because the state-of-the-art algorithms for computing lower bounds \bm and \ls are optimized for this case. Note, that \forilp can be straightforwardly extended to an "anchor-aware" version, to compute a lower bound when a subset of the vertices is already mapped, by fixing $x_{i,k} = 1$ if $i\in V_{G+\varepsilon}$ is mapped to $k \in V_{H+\varepsilon}$, as done in \cite{ChangFYQZ23} with for the \bm lower bound. 

Figure~\ref{fig:gaps} shows the mean and maximum gaps in percentage (log-scaled) for the query graphs on the x-axis, computed on the whole dataset w.r.t. the GED values, while Figure~\ref{fig:lb_runtimes} illustrates the corresponding average runtimes in milliseconds per pair $(Q,H)$. The gap of a lower bound algorithm \textsc{alg} on a graph pair $Q,H$ is computed as $\frac{GED(Q,H) - \textsc{alg}(Q,H)}{GED(Q,H)}$. Optimal GEDs were computed using \fori.

As expected from our theoretical discussion, there is a clear hierarchy of the lower bounds both in terms of quality of the bound and computational efficiency. 
\forilp offers the tightest but slowest to compute lower bound, the \ls bound is the fastest to compute but offers the loosest bound, while \bm stands between the two other heuristics both in terms of runtime and quality of the lower bound. 
The lower bound of \forilp offers a tremendous improvement in terms of quality compared to the other algorithms. Its mean gap is an order of magnitude smaller than that of \bm for all but one query graph on \aids and every query graph on \muta, with similar behavior of the max gaps.
In particular, maximum gaps achieved by \forilp are almost always smaller than both the \ls and \bm average gaps. 
However, computing lower bounds with \forilp incurs a non-negligible runtime cost. In contrast, \ls demonstrates exceptional efficiency, achieving average runtimes in the hundreds of microseconds. Similarly, \bm maintains low computational overhead, consistently staying below 10 milliseconds. Meanwhile, \forilp typically operates at an order of magnitude higher, making it the most time-consuming among the three.

Nonetheless, the impressive quality of the \forilp lower bound lets us expect a growing impact on the runtime in the graph similarity search as $\tau$ increases, as the \bm heuristic will filter out less and less graphs for \bmaoged before needing to compute the GED.

\begin{figure}
    \centering
    \begin{subfigure}[b]{\linewidth}
      \centering
        \includegraphics[width=\linewidth]{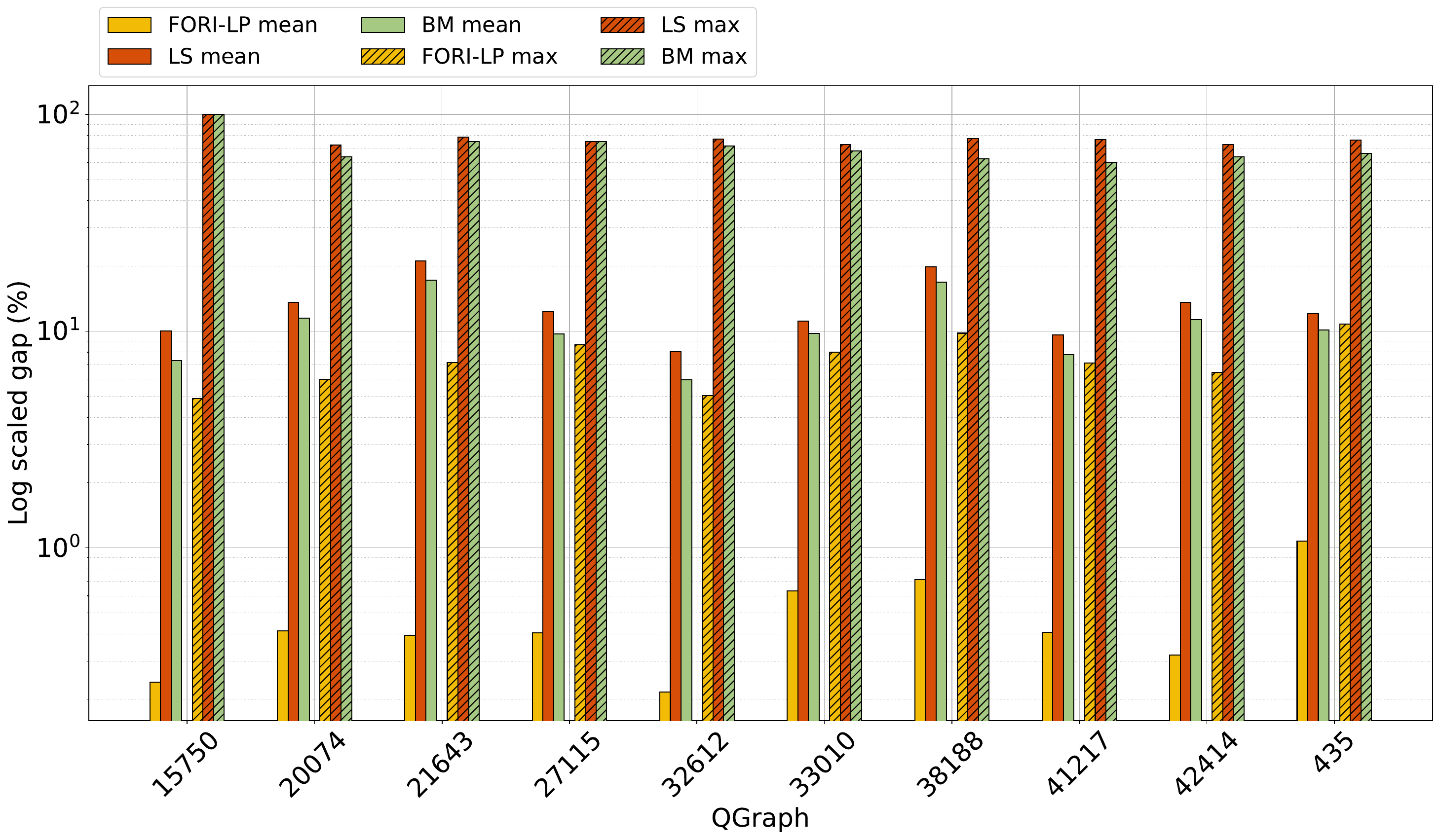}
    \caption{\aids}
    \label{fig:gap_aids}
  \end{subfigure}
  \\  \vspace{2mm}
  \begin{subfigure}[b]{\linewidth}
      \centering
        \includegraphics[width=\linewidth]{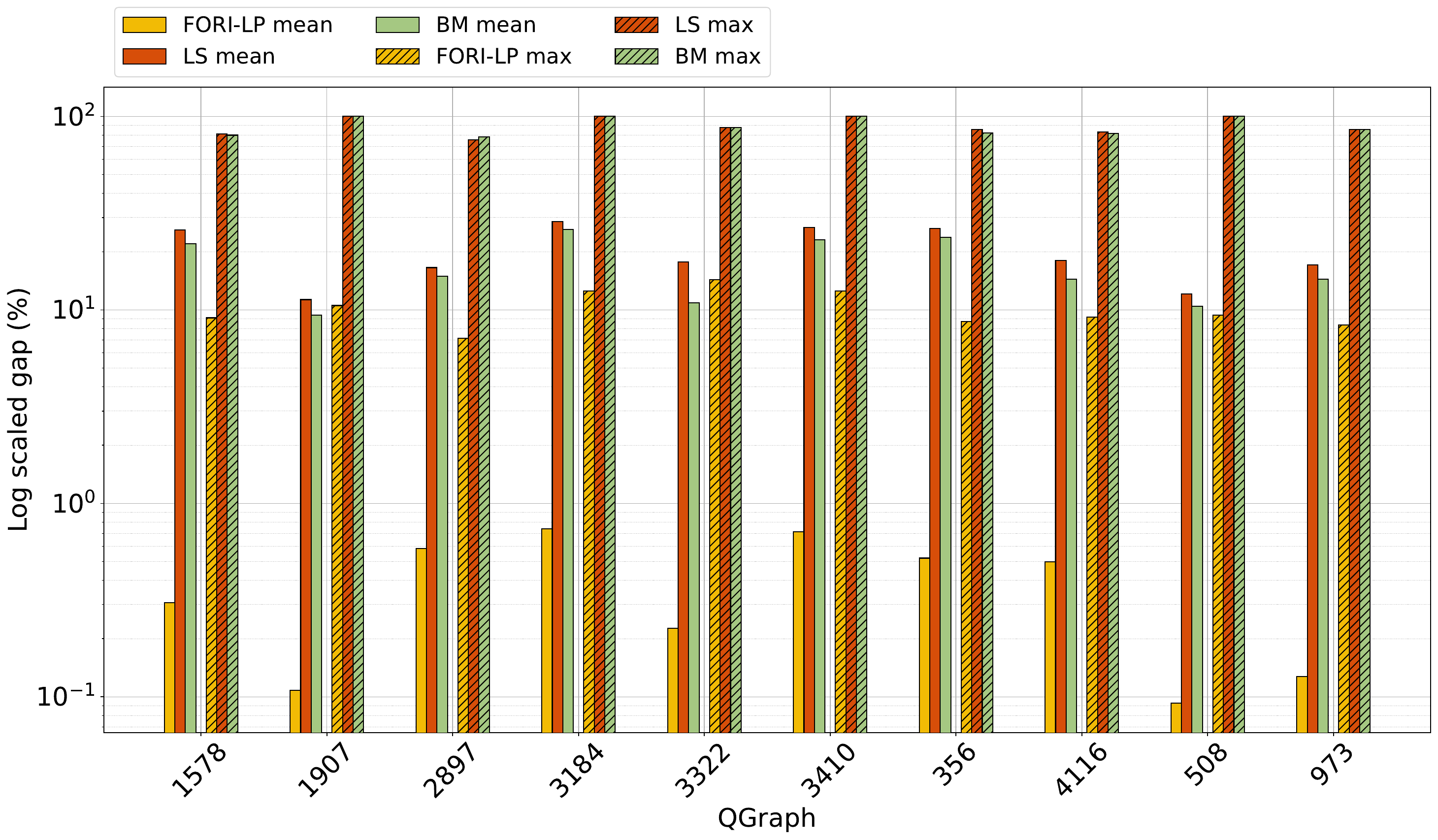}
        \caption{\muta}
        \label{fig:gap_muta}
  \end{subfigure}
  \caption{Lower bound comparison between \ls, \bm, and \forilp grouped by query graphs and divided by datasets. }
  \label{fig:gaps}
\end{figure}

\begin{figure}
    \centering
    \begin{subfigure}[b]{\linewidth}
      \centering
        \includegraphics[width=\linewidth]{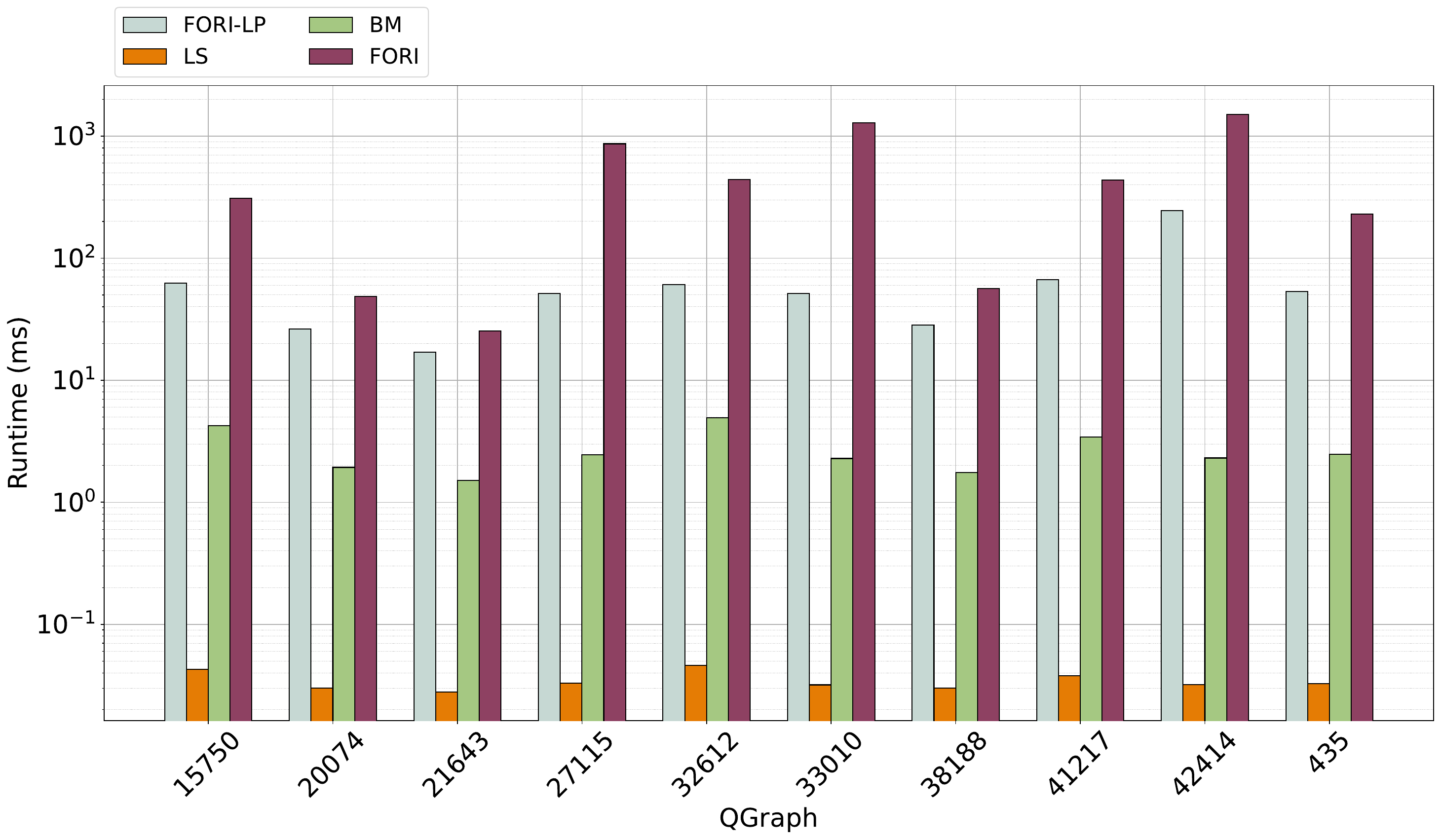}
    \caption{\aids}
    \label{fig:lb_runtimes_aids}
  \end{subfigure}
  \\  \vspace{2mm}
  \begin{subfigure}[b]{\linewidth}
      \centering
        \includegraphics[width=\linewidth]{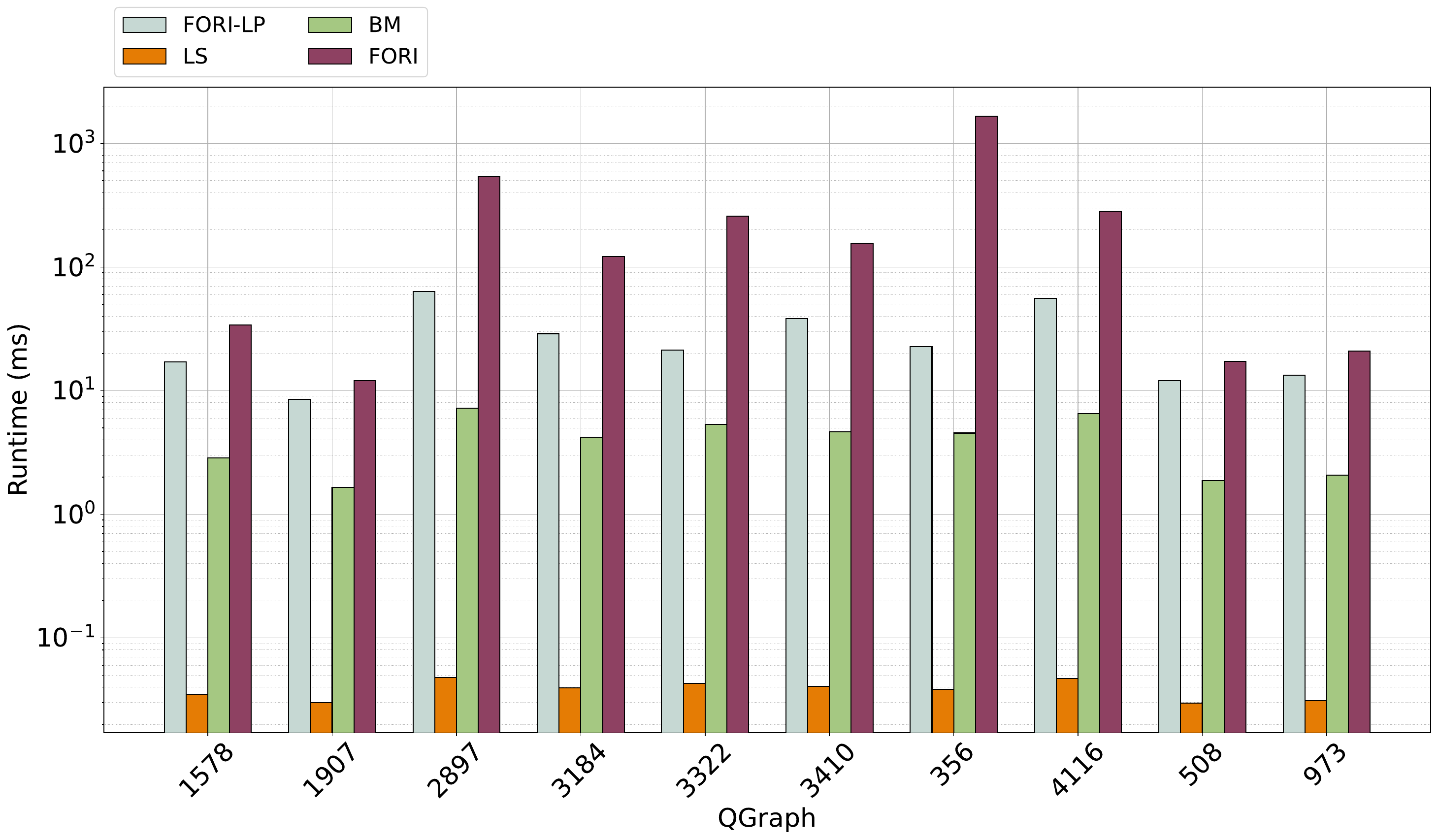}
    \caption{\muta}
    \label{fig:lb_runtimes_muta}
  \end{subfigure}
  \caption{Average runtime comparison between \ls, \bm, \forilp, and \fori grouped by query graphs and divided by datasets.}
  \label{fig:lb_runtimes}
\end{figure}

\subsection*{\textbf{Answering Q2:} Comparison between \forisim and \bmaoged}
In this section, we evaluate the performance of \forisim "in" graph similarity search, comparing it against \bmaoged, the current state-of-the-art approach~\cite{ChangFYQZ23}.
For this comparison, we restrict to unit edit costs, as \bmaoged does not support non-uniform edit costs.
As a consequence, we use only the \aids and \muta datasets for the comparison, as their original cost function can be naturally translated into the unit cost model (see, e.g., \cite{BlumenthalBGBB20,DAscenzoMMR25}).

Concerning the similarity threshold, we examine the parameter $\tau$ over the set $\{1,5\} \cup [10,20] \cup \{30,40,50\}$.
As illustrated in Figures~\ref{fig:fraction_aids}--\ref{fig:fraction_muta}, even high threshold values such as 50 lead to a substantial portion of dataset graphs being filtered out. For instance, in the \aids dataset, fewer than 10\% of the graphs have a GED of 50 or less with respect to query graphs 32612, 15750, and 41217. Similarly, for the \muta dataset, query graph 2897 yields a GED below 50 with fewer than 30\% of the graphs in the dataset.
To the best of our knowledge, this is the first study to explore threshold values exceeding 16~\cite{GoudaH16,KimC019,ChangF0QZO20,ChangFYQZ23}. 

\begin{figure}
    \centering
    \begin{subfigure}[b]{.8\linewidth}
      \centering
        \includegraphics[width=\linewidth]{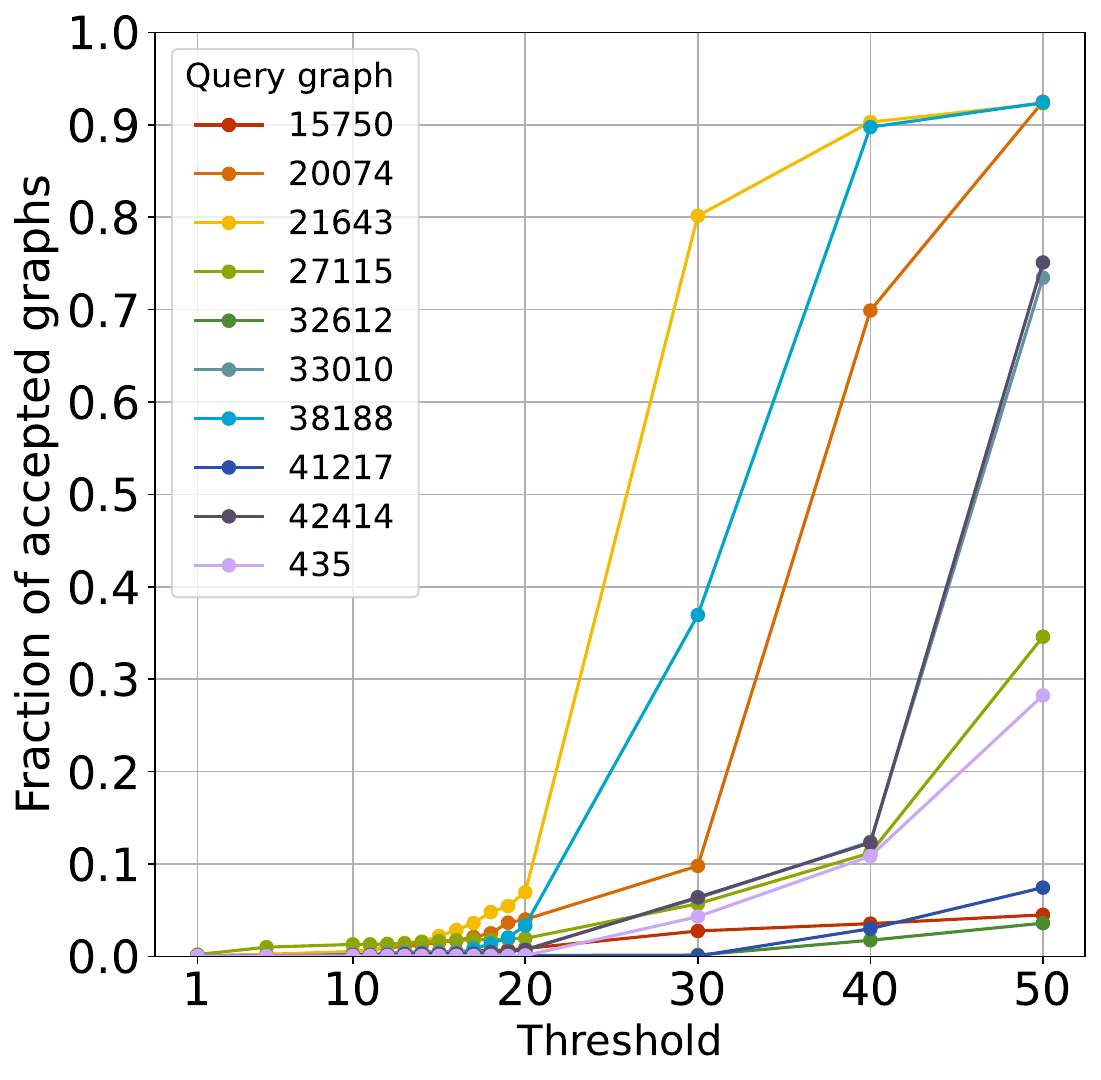}
        \caption{\aids}
        \label{fig:fraction_aids}
  \end{subfigure}
  \\  \vspace{2mm}
  \begin{subfigure}[b]{.8\linewidth}
      \centering
        \includegraphics[width=\linewidth]{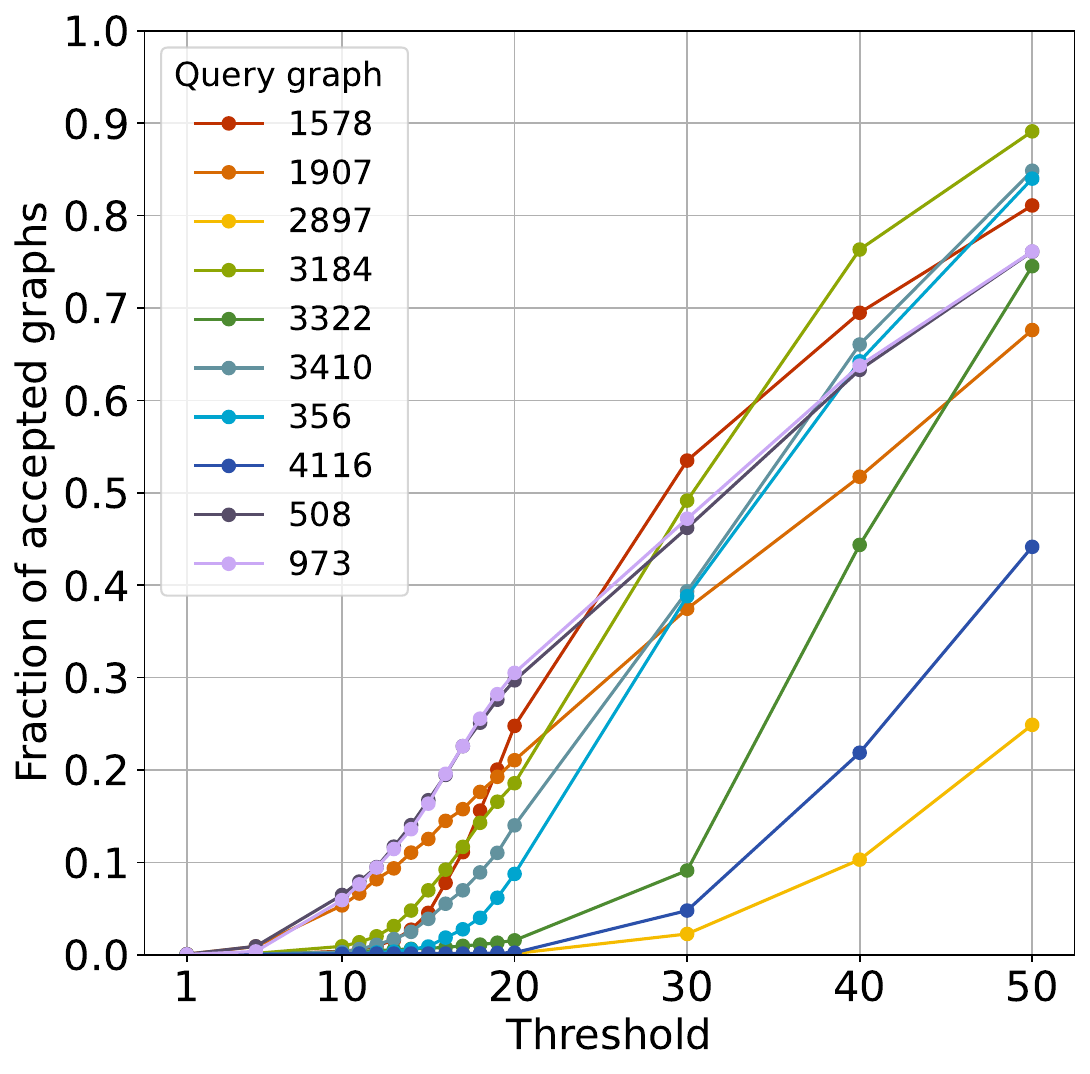}
        \caption{\muta}
        \label{fig:fraction_muta}
  \end{subfigure}
  \caption{Fraction of average accepted graphs over all query graphs on unit edit costs.}
  \label{fig:fraction}
\end{figure}

Table~\ref{tab:aggr_gss} presents the statistics computed over the query graphs for each dataset, grouped by threshold~$\tau$. A time limit of one hour is set for each query graph. When aggregating results by $\tau$, we use \tlr to indicate that at least one query graph instance exceeded the time limit. The column \textit{Matches} reports the average number of dataset graphs accepted. Column \textbf{T} shows the average runtime per query graph, while column \textbf{M} indicates the maximum memory usage, i.e., the "maximum resident set size
of the process during its lifetime", as measured by the
GNU command \textit{time}\footnote{\url{https://man7.org/linux/man-pages/man1/time.1.html}}. Lastly, column \textbf{C} represents the percentage of dataset graphs in $\mathcal{D}$ evaluated by each algorithm within the time limit.
            
\begin{table}[ht]
    \centering
    \caption{Results statistics for graph similarity search with unit edit costs on 10 query graphs.}
    \label{tab:aggr_gss}
        \setlength{\tabcolsep}{4pt}

    \begin{subtable}[t]{\columnwidth}
        \centering
        \caption{\aids}
        \renewcommand{\arraystretch}{1.1}
        \footnotesize
        \begin{tabular}{cc|ccc|ccc}
            \hline
            \multirow{2}{*}{$\tau$} & \multirow{2}{*}{Matches} 
            & \multicolumn{3}{c|}{\textbf{\forisim}} 
            & \multicolumn{3}{c}{\textbf{\bmaoged}} \\
            & & \textbf{T (s)} & \textbf{M} & \textbf{C (\%)} 
              & \textbf{T (s)} & \textbf{M} & \textbf{C (\%)} \\
            \hline
            1   & 1.3 ($<$0.1\%)    & 0.14     & 261M   & 100   & $<$0.01      & 5M    & 100   \\
            5   & 3.5 (0.2\%)   & 0.48     & 373M   & 100   & 0.15    & 6M    & 100   \\
            10  & 6.0 (0.3\%)   & 1.61     & 416M   & 100   & 60.67  & 379M   & 100   \\
            15 & 13.9 (0.7\%) & 4.55 & 567M & 100 &  \tlr   & 1.5G & 95.9 \\
            20  & 37.6 (1.9\%)   & 11.70     & 657M   & 100   & \tlr      & 8.1G  & 67.8  \\
            30  & 305.1 (15.3\%)  & 37.67    & 851M   & 100   & \tlr      & 15.2G  & 27.5  \\
            40  & 610.0 (30.5\%) & 61.35    & 1G     & 100   & \tlr      & 10.7G  & 18.9  \\
            50  & 1008.3 (50.4\%) & 98.32    & 1.5G   & 100   & \tlr      & 10.7G  & 13.4  \\
            \hline
        \end{tabular}
    \end{subtable}
    \\  \vspace{2mm}
    \begin{subtable}[t]{\columnwidth}
        \centering
        \caption{\muta}
        \renewcommand{\arraystretch}{1.1}
        \footnotesize
        \begin{tabular}{cc|ccc|ccc}
            \hline
            \multirow{2}{*}{$\tau$} & \multirow{2}{*}{Matches} 
            & \multicolumn{3}{c|}{\textbf{\forisim}} 
            & \multicolumn{3}{c}{\textbf{\bmaoged}} \\
            & & \textbf{T (s)} & \textbf{M} & \textbf{C (\%)} 
              & \textbf{T (s)} & \textbf{M} & \textbf{C (\%)} \\
            \hline
            1   & 1.6  ($<$ 0.1\%)  & 0.10     & 315M   & 100   & $<$0.01      & 7M    & 100   \\
            5   & 11.5  (0.3\%)  & 1.56     & 412M   & 100   & 0.03    & 8M    & 100   \\
            10  & 86.5 (2.0\%)    & 13.89     & 511M   & 100   & 4.97  & 13M   & 100   \\
            15  & 272.9 (6.3\%) & 30.10     & 665M      & 100  & 471.70 & 645M & 100 \\
            20  & 648.0 (14.9\%)  & 63.88     & 797M   & 100   & \tlr      & 3.6G  &  47.9 \\
            30  & 1422.5 (32.8\%) & 157.65    & 1.1G   & 100   & \tlr      & 19.4G  & 40.3  \\
            40  & 2306.0 (53.1\%)  & 257.27    & 1.4G     & 100   & \tlr      & 13.1G  & 35.7  \\
            50  & 3048.2 (70.2\%) & 296.20    & 1.7G   & 100   & \tlr      & 4.2G  & 13.7G  \\
            \hline
        \end{tabular}
    \end{subtable}
\end{table}

\begin{figure}
    \centering
    \begin{subfigure}[b]{.8\linewidth}
      \centering
        \includegraphics[width=\linewidth]{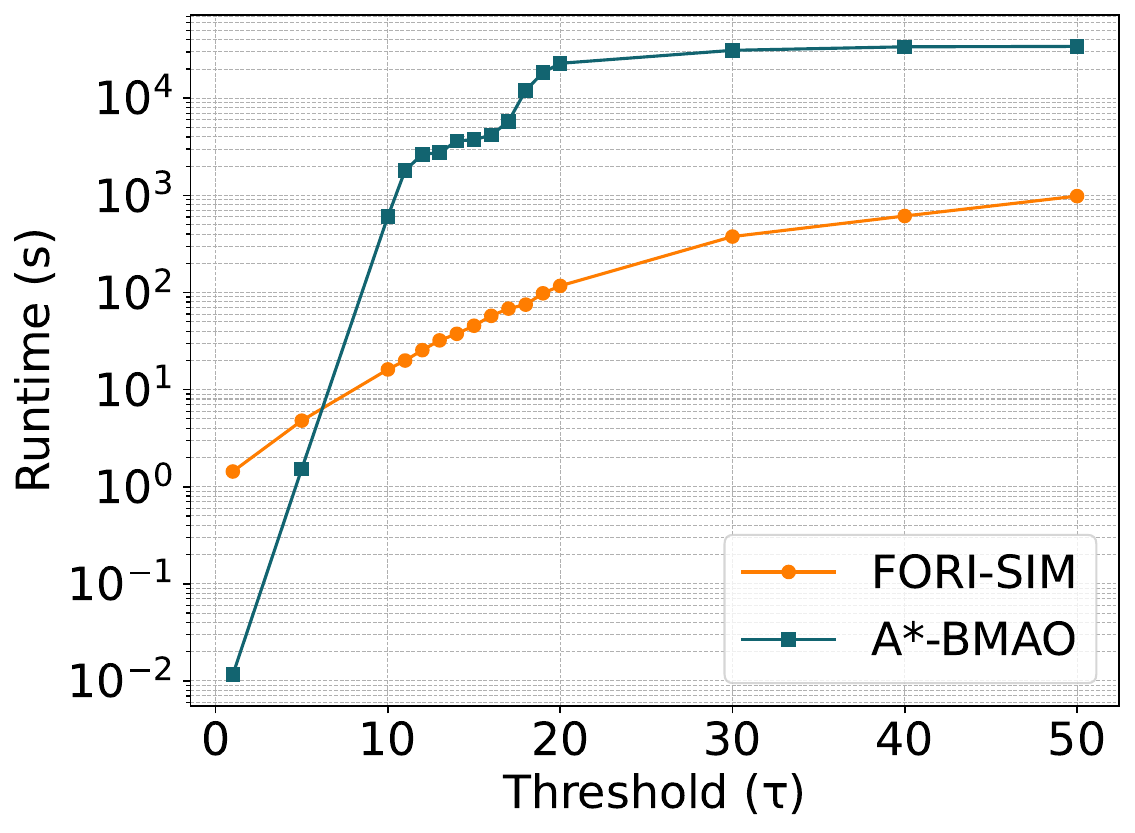}
        \caption{\aids}
        \label{fig:lineplot_aids}
  \end{subfigure}
  \\  \vspace{2mm}
  \begin{subfigure}[b]{.8\linewidth}
      \centering
        \includegraphics[width=\linewidth]{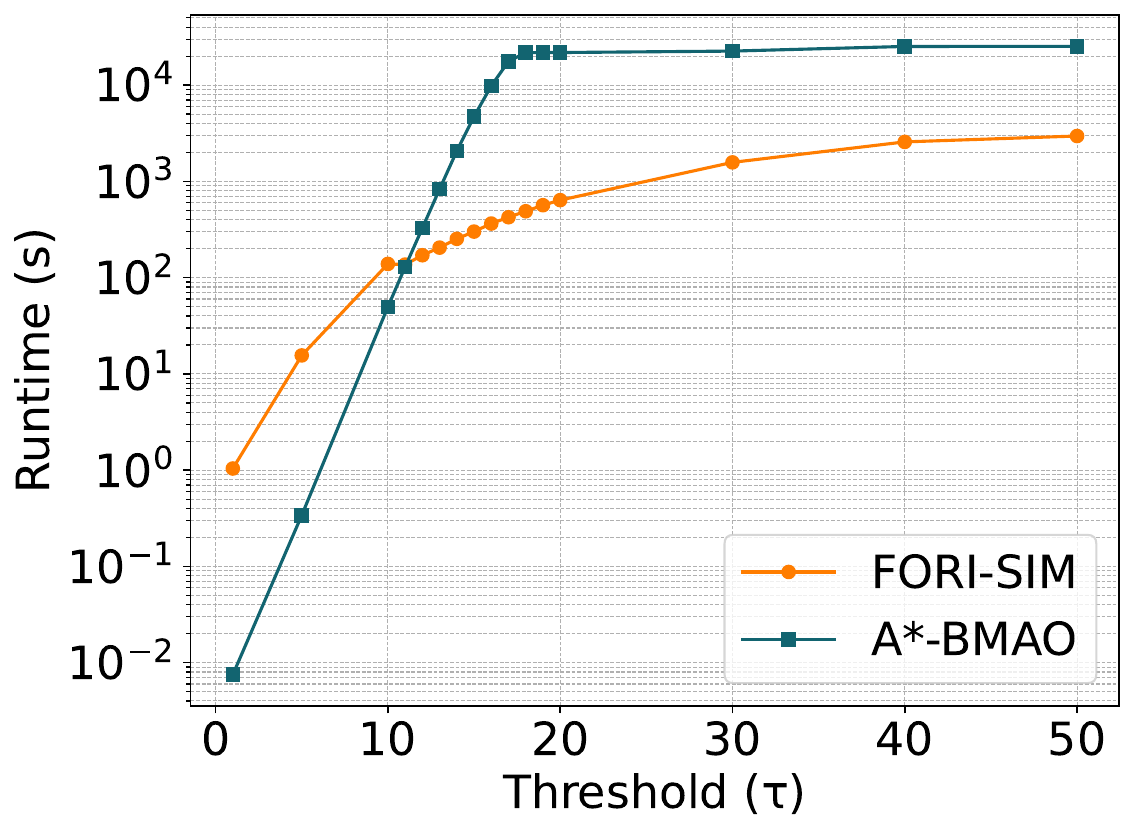}
        \caption{\muta}
        \label{fig:lineplot_muta}
  \end{subfigure}
  \caption{Line plots of runtime performance of \forisim and \bmaoged on the graph similarity search problem as a function of $\tau$.}
  \label{fig:lineplots}
\end{figure}

\noindent\textbf{\aids dataset.}
Both algorithms perform efficiently at low thresholds. For $\tau = 1$ and $\tau = 5$, \bmaoged is notably faster and more memory-efficient than \forisim, completing queries in under 0.15 seconds and using less than 10MB of memory. However, this advantage diminishes rapidly as the threshold increases. At $\tau = 10$, \bmaoged's runtime exceeds 60 seconds, while \forisim remains stable at just 1.61 seconds. Beyond this point, \bmaoged consistently fails to complete queries within the time limit, with \tlr appearing from $\tau = 15$ onward. This failure is accompanied by a steep rise in memory consumption, reaching 15.2GB at $\tau = 30$, and a dramatic drop in coverage from 95.9\% at $\tau = 15$ to just 13.4\% at $\tau = 50$. The one-hour time limit is further justified by these performance trends, as allowing \bmaoged to run to completion would likely require several hours and a potentially larger memory footprint.
In contrast, \forisim maintains full coverage and completes all queries within the time limit, even at the highest thresholds. Its memory usage increases gradually, peaking at 1.5GB, and its runtime remains manageable, staying under 100 seconds at $\tau = 50$. These results suggest that while \bmaoged is highly efficient for small similarity thresholds, it lacks the scalability required for larger thresholds, whereas \forisim offers consistent and reliable performance.

\noindent\textbf{\muta dataset.}
At low thresholds, \bmaoged again shows impressive speed and minimal memory usage, completing queries in milliseconds. However, its performance deteriorates fast. By $\tau = 15$, its runtime raises to nearly 472 seconds, and from $\tau = 20$ onward, it fails to complete queries within the time limit. Memory usage also becomes a critical bottleneck, peaking at 19.4GB at $\tau = 30$. Coverage drops fast, falling below 50\% at $\tau = 20$ and reaching just 35.4\% at $\tau = 50$. \forisim, on the other hand, demonstrates remarkable scalability. It completes all queries within the time limit across all thresholds, with runtime increasing monotonically from 0.10 seconds at $\tau = 1$ to 296.20 seconds at $\tau = 50$. Memory usage remains within reasonable bounds, peaking at 1.7GB, and coverage stays consistently at 100\%. The number of matches also grows steadily with increasing $\tau$ in both datasets, reflecting the relaxed similarity constraints.

Figure~\ref{fig:lineplots} compares the runtime performance of \forisim and \bmaoged as a function of the threshold parameter $\tau$.
The line plots report the total runtime (in seconds, log-scale) over all the 10 query graphs, grouped by thresholds.
As the threshold increases, the total runtime for \forisim (depicted by the orange curve) exhibits a smooth and monotonic growth. This trend appears nearly linear when viewed in log-scale, starting from fractions of a second at low thresholds and reaching approximately $10^3$ seconds at $\tau = 50$. The consistency of this growth suggests that our method scales predictably and remains computationally feasible even at high similarity tolerances.

In contrast, \bmaoged (shown in blue) displays a different behavior. Its runtime rises sharply, with an exponential increase up to thresholds around $\tau = 15$–$20$. Beyond this point, the curve flattens near $10^4$ seconds. This saturation indicates that \bmaoged reaches the time limit and struggles to process queries efficiently at higher thresholds. Although \bmaoged performs slightly faster than \forisim for $\tau \leq 10$, the gap reverses dramatically as $\tau$ increases, with the latter becoming significantly more efficient.
These runtime trends underscore the scalability advantages of our routine. While both algorithms experience increased computational demand with rising thresholds, \forisim maintains a controlled growth, whereas \bmaoged suffers from combinatorial explosion of the search space.

\subsection*{\textbf{Answering Q3:}
\begin{table}[ht]
    \centering
    \caption{Results statistics for graph similarity search with non-uniform edit costs on 10 query graphs.}
    \label{tab:aggr_gss_nonuni}
    \footnotesize
    \begin{subtable}[t]{\columnwidth}
        \centering
        \caption{\aids}
        \label{tab:aggr_gss_nonuni_aids}
        \renewcommand{\arraystretch}{1.2}
        \begin{tabular}{cc|ccc}
            \hline
            \multirow{2}{*}{$\tau \times 3.575$} & \multirow{2}{*}{Matches} 
            & \multicolumn{3}{c}{\textbf{\forisim}}\\
            & & \textbf{T (s)} & \textbf{M} & \textbf{C (\%)}\\
            \hline
            1   & 1.1 ($<$ 0.1\%)    & 0.25     & 308M   & 100   \\
            5   & 6.8  (0.3\%) & 4.71     & 515M   & 100   \\
            10  & 49.5 (2.4\%)   & 20.24     & 851M   & 100   \\
            15 & 329.0 (16.4\%) & 46.17 & 1.0G & 100 \\
            20  & 581.6 (29.0\%)   & 55.59     & 1.1G   & 100  \\
            30  & 1319.8 (66.0\%)  & 77.98    & 1.4G   & 100   \\
            40  & 1563.3 (78.2\%)  & 74.33    & 1.4G     & 100   \\
            50  & 1746.7 (87.3\%) & 68.19    & 1.2G   & 100   \\
            \hline
        \end{tabular}
    \end{subtable}
    \\  \vspace{2mm}
    \begin{subtable}[t]{\columnwidth}
        \centering
        \caption{\muta}
        \label{tab:aggr_gss_nonuni_muta}
        \renewcommand{\arraystretch}{1.2}
        \begin{tabular}{cc|ccc}
            \hline
            \multirow{2}{*}{$\tau \times 3.575$} & \multirow{2}{*}{Matches} 
            & \multicolumn{3}{c}{\textbf{\forisim}} \\
            & & \textbf{T (s)} & \textbf{M} & \textbf{C (\%)}\\
            \hline
            1   &   1.7 ($<$ 0.1\%) &   0.43    &   367M    &   100 \\
            5   &   39.7 (0.9\%)  &   24.94   &   654M    &   100 \\
            10  &   471.8 (10.9\%)  &   86.20  & 1.0G  &   100\\
            15  &   1273.8 (29.4\%)  &   163.23  &   1.5G & 100 \\
            20  &   2119.3 (48.8\%)  &   220.06  &   1.7G & 100 \\
            30  &   3447.3 (79.4\%) &   230.67  &   1.9G & 100 \\
            40  &   3997.9 (92.1\%) &   200.65  &   1.9G & 100  \\
            50 &    4193.6 (96.6\%) &   187.00  &   2.2G & 100\\
            \hline
        \end{tabular}
    \end{subtable}
    \\  \vspace{2mm}
    \begin{subtable}[t]{\columnwidth}
        \centering
        \caption{\protein}
        \label{tab:aggr_gss_nonuni_prot}
        \renewcommand{\arraystretch}{1.2}
        \begin{tabular}{cc|ccc}
            \hline
            \multirow{2}{*}{$\tau \times 8.375$} & \multirow{2}{*}{Matches} 
            & \multicolumn{3}{c}{\textbf{\forisim}} \\
            & & \textbf{T (s)} & \textbf{M} & \textbf{C (\%)}\\
            \hline
            1 & 1.1 (0.1\%) & 63.83 & 4.1G & 100 \\
            5 & 1.4 (0.2 \%) & 65.25 & 3.8G & 100 \\
            10 & 2.1 (0.3\%) & 64.69 & 3.7G & 100 \\
            15 & 2.7 (0.4 \%) & 65.90 & 3.4G & 100 \\
            20 & 23.5 (3.9\%) & 66.35 & 5.0G & 100 \\
            30 & 138.4 (23.0\%) & 69.79 & 4.5G & 100 \\
            40 & 353.2 (58.9\%) & 75.69 & 5.2G & 100 \\
            50 & 485.3 (81.0\%)& 77.16 & 4.6G & 100 \\
            \hline
        \end{tabular}
    \end{subtable}
\end{table}GED verification on non-uniform edit costs}
In this section we report on a computational experience for graph similarity search under non-uniform edit costs.
We use the cost functions described in~\cite{BlumenthalBGBB20}.
For \aids and \muta datasets node edit costs are defined as $c_V(\alpha,\alpha') := 5.5 \cdot \delta_{\alpha \ne \alpha'}$, $c_V(\alpha,\epsilon) := 2.75$, and $c_V(\epsilon, \alpha'):=2.75$, for all $(\alpha,\alpha') \in \Sigma_V\times\Sigma_V$. Edge edit costs are defined as $c_E(\beta,\beta'):= 1.65 \cdot \delta_{\beta,\beta'}$, $c_E(\beta,\epsilon):=0.825$, and $c_E(\epsilon, \beta'):=0.825$, for all $(\beta, \beta') \in \Sigma_E \times \Sigma_E$.
Turning to \protein dataset, node edit costs are defined as: $c_V(\alpha, \alpha') := 16.5 \cdot \delta_{\alpha.t \neq \alpha'.t} + 0.75 \cdot \delta_{\alpha.t = \alpha'.t} \cdot \text{LD}(\alpha.s, \alpha'.s)$, $c_V(\alpha, \varepsilon) := 8.25$, and $c_V(\varepsilon, \alpha') := 8.25$, for all $(\alpha, \alpha') \in \Sigma_V \times \Sigma_V$, where $\text{LD}(\cdot, \cdot)$ denotes the Levenshtein string edit distance.
Edge edit costs are defined as: $c_E(\beta, \beta') := 0.25 \cdot \text{LSAPE}(C_{\beta, \beta'})$, $c_E(\beta, \varepsilon) := 0.25 \cdot f(\beta)$, and $c_E(\varepsilon, \beta') := 0.25 \cdot f(\beta')$, for all $(\beta, \beta') \in \Sigma_E \times \Sigma_E$, where $f(\beta) := 1 + \delta_{\beta.t \neq \texttt{null}},$
and $C_{\beta, \beta'} \in \mathbb{R}^{(f(\beta)+1) \times (f(\beta')+1)}$ is constructed as $c^{\beta, \beta'}_{r,s} := 2 \cdot \delta_{\beta.t_r \neq \beta'.t_s}$, $c^{\beta, \beta'}_{r, f(\beta')+1} := 1$, $c^{\beta, \beta'}_{f(\beta)+1, s} := 1$,
for all $(r, s) \in [f(\beta)] \times [f(\beta')]$, and $\text{LSAPE}(C_{\beta, \beta'})$ is the cost of an optimal solution to the \textit{linear sum assignment problem with error-correction}~\cite{BlumenthalBGBB20}.

Concerning the similarity threshold, we examine the parameter $\tau$ over the values in the set $\{1,5\} \cup [10,20] \cup \{30,40,50\}$ multiplied by a dataset-specific constant. For \aids and \muta such constant is equal to 3.575, i.e., the average cost of substituting a node and an edge. For \protein it is computed similarly and is equal to 8.375.
As illustrated in Figures~\ref{fig:fraction_aids_nonuni}--\ref{fig:fraction_prot_nonuni}, the chosen thresholds allow for a comprehensive evaluation of \forisim performances, from cases where most graphs are rejected at low thresholds, to near-complete acceptance at high thresholds, as well as the gradual transitions in between, as seen in Figure~\ref{fig:fraction} for unit edit costs.

In both analysis we will refer to the threshold values without the constant multiplicative factor for the ease of readability, as reported in Tables~\ref{tab:aggr_gss_nonuni_aids}--~\ref{tab:aggr_gss_nonuni_prot}.

Since the \ls heuristic works only for unit edit costs, it is not possible to use it in the non-uniform costs setting of this section.
On \protein, due to the definition of its edit cost function, the runtime of the \bm lower bound is $O((|V_G|+|V_H|)^3)$ which is close to the one of \forilp. We therefore chose to skip its computation in \forisim on this dataset.
For this reason, we analyze the \aids and \muta datasets jointly, whereas the \protein dataset is discussed independently.

\begin{figure}
    \centering
    \begin{subfigure}[b]{0.7\linewidth}
      \centering
        \includegraphics[width=\linewidth]{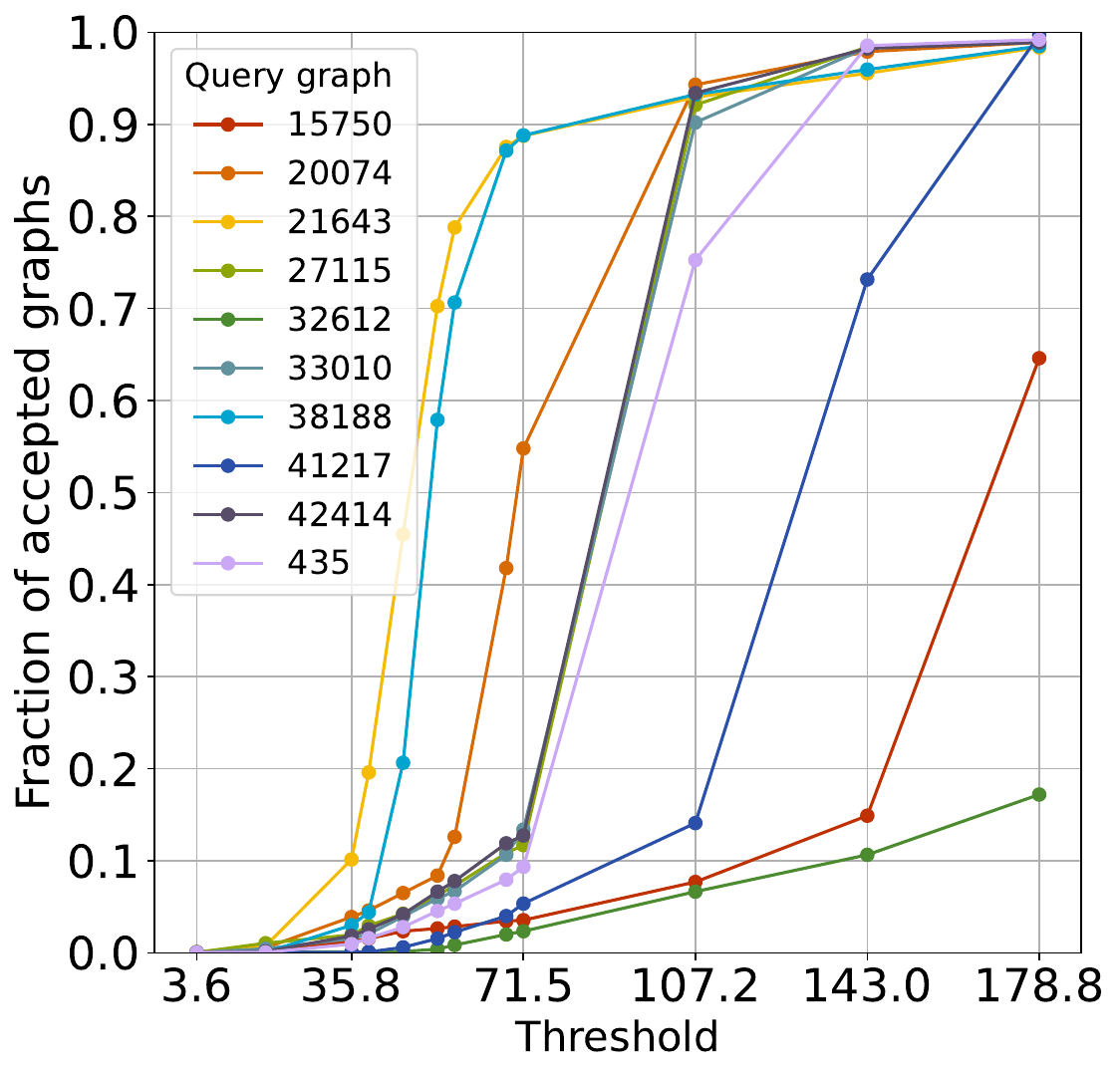}
        \caption{\aids}
        \label{fig:fraction_aids_nonuni}
  \end{subfigure}
  \\ \vspace{2mm}
  \begin{subfigure}[b]{0.7\linewidth}
      \centering
        \includegraphics[width=\linewidth]{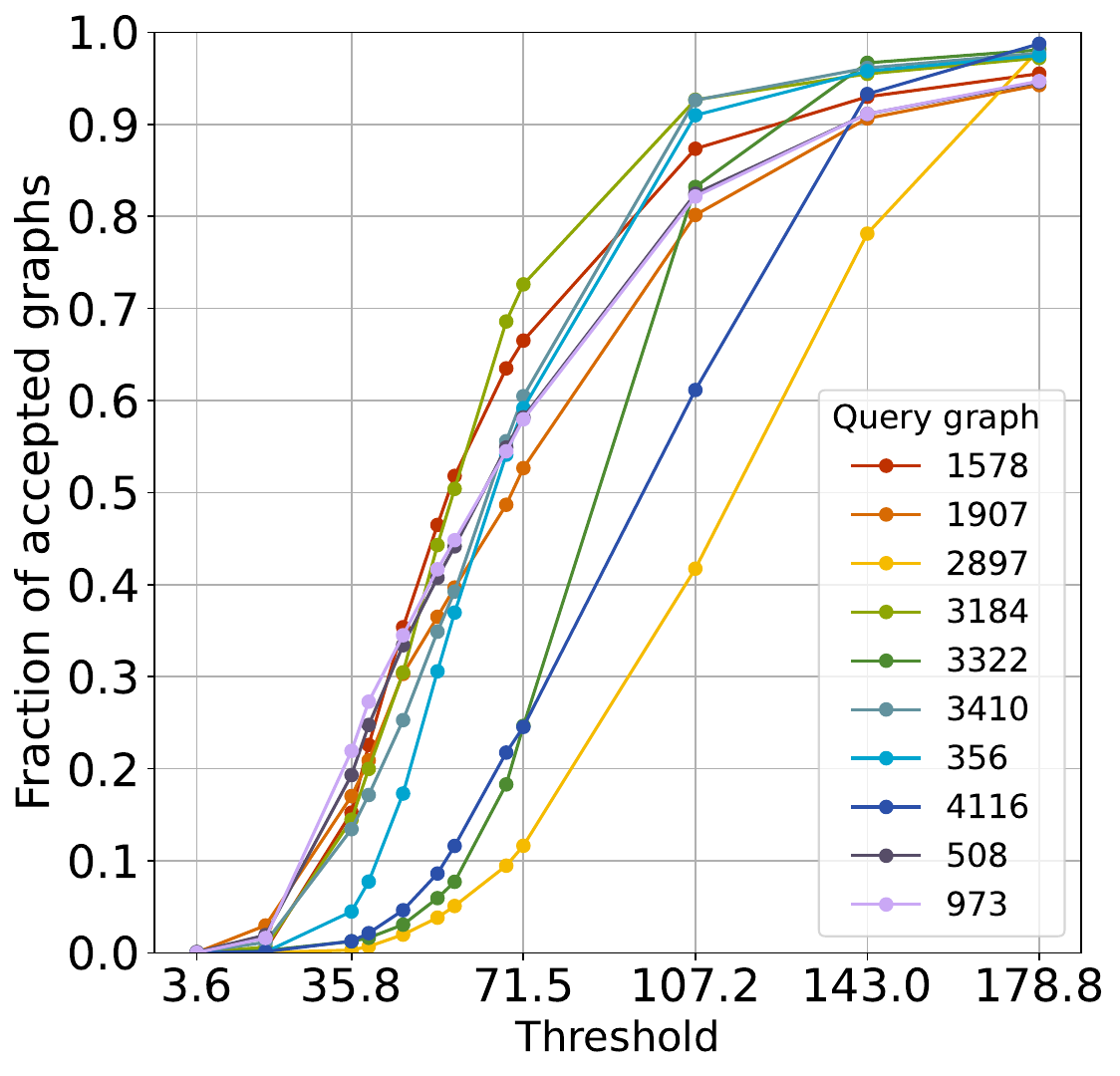}
        \caption{\muta}
        \label{fig:fraction_muta_nonuni}
  \end{subfigure}
  \\ \vspace{2mm}
  \begin{subfigure}[b]{0.7\linewidth}
      \centering
        \includegraphics[width=\linewidth]{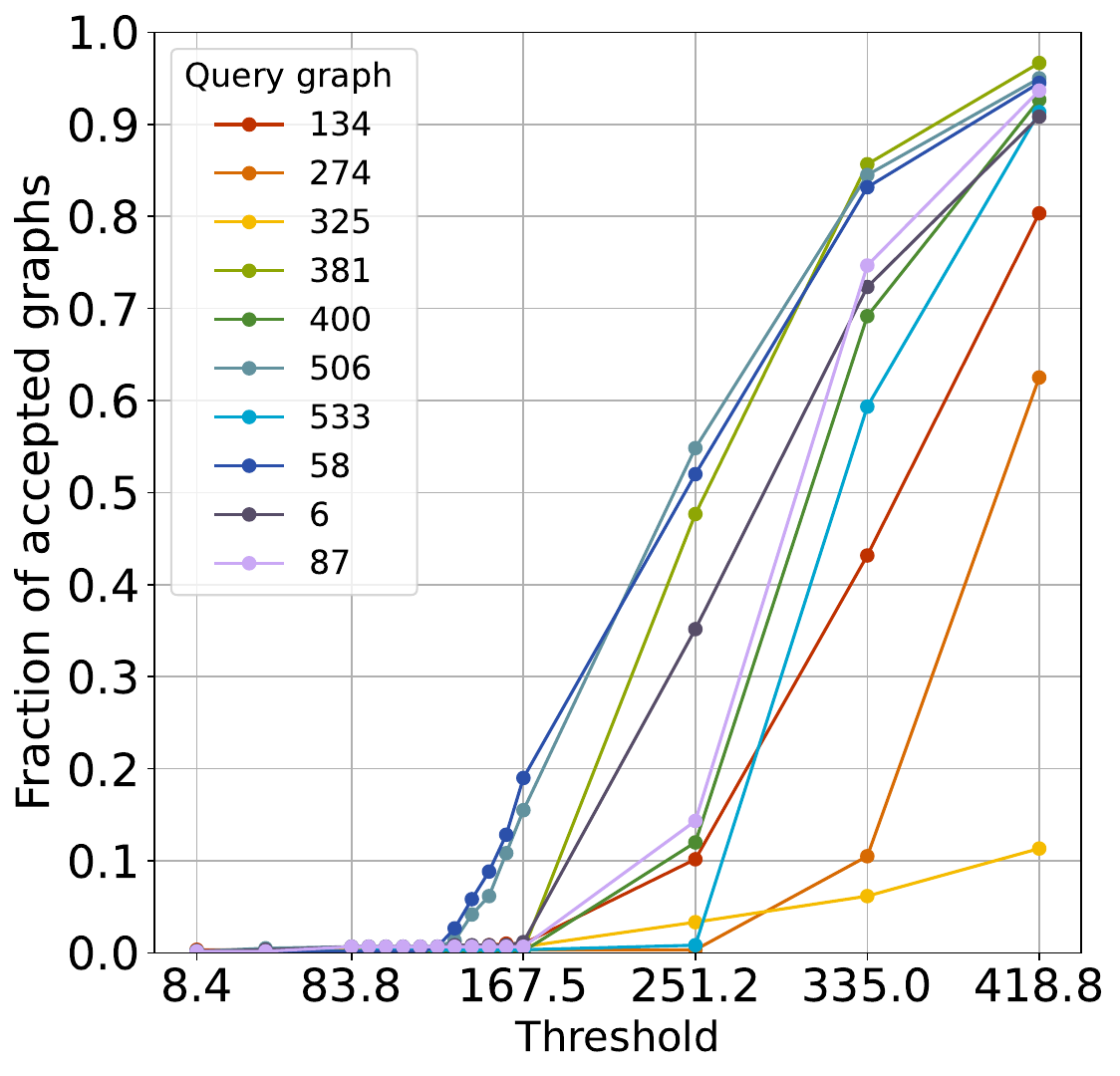}
        \caption{\protein}
        \label{fig:fraction_prot_nonuni}
  \end{subfigure}
  \caption{Fraction of average accepted graphs over all query graphs with non-uniform edit costs.}
  \label{fig:fraction_nonuni}
\end{figure}

\noindent\textbf{\aids and \muta datasets.} As expected, the average number of accepted graphs grows steadily with the threshold, reaching up to 87.3\% and 96.6\% of the dataset for \aids and \muta respectively at $\tau = 50$. Runtime increases accordingly, with \aids peaking at 77.98 seconds and \muta at 230.67 seconds, both reached at $\tau=30$. At higher thresholds, a minor drop is observed, likely attributable to the solver's reduced computational effort in finding feasible solutions (see also the aggregated runtime trends in Figures~\ref{fig:lineplot_aids_nonuni})--~\ref{fig:lineplot_muta_nonuni}). Memory usage scales with threshold, reaching a maximum of 2.2G for \muta and 1.4G for \aids, indicating a heavier computational load on the former. As in the unit edit cost case, the coverage remains consistently at 100\% across all thresholds, confirming that \forisim successfully evaluates all dataset graphs within the time limit. These results demonstrate the scalability and robustness of \forisim, even under non-uniform edit cost functions.

\noindent\textbf{Protein dataset.} 
The results on the \protein dataset reveal a distinct runtime profile for \forisim, primarily due to the inapplicability of the \ls and \bm lower bounds to its edit cost function.
As a consequence, all instances are evaluated using our \forilp and \forithr only, resulting in a relatively stable runtime across the entire range of thresholds (cf. Figure~\ref{fig:lineplot_prot_nonuni}).
In detail, from $\tau=1$ to $\tau = 20$, the average runtime per query graph remains clustered around 64-66 seconds, despite a gradual increase in the number of accepted graphs.
This stability reflects the uniform computational effort required by \forisim in the absence of fast lower bounds computation.
Beyond $\tau=20$, the number of accepted graphs increases significantly, reaching 81.0\% at $\tau=50$, and the runtime begins to rise accordingly, peaking at 77.16 seconds.
Memory usage ranges between 3.4G and 5.2G, with no clear trend, suggesting that memory demands are influenced more by the query graph structure than the threshold.
Coverage remains consistently at 100\% across all thresholds, confirming that \forisim is able to consider all dataset graphs within the time limit.
Overall, results on \protein dataset highlight the robustness of \forisim, and its performance stability even when fast lower bound algorithms are not available.

\begin{figure}
    \centering
    \begin{subfigure}[b]{.8\linewidth}
      \centering
        \includegraphics[width=\linewidth]{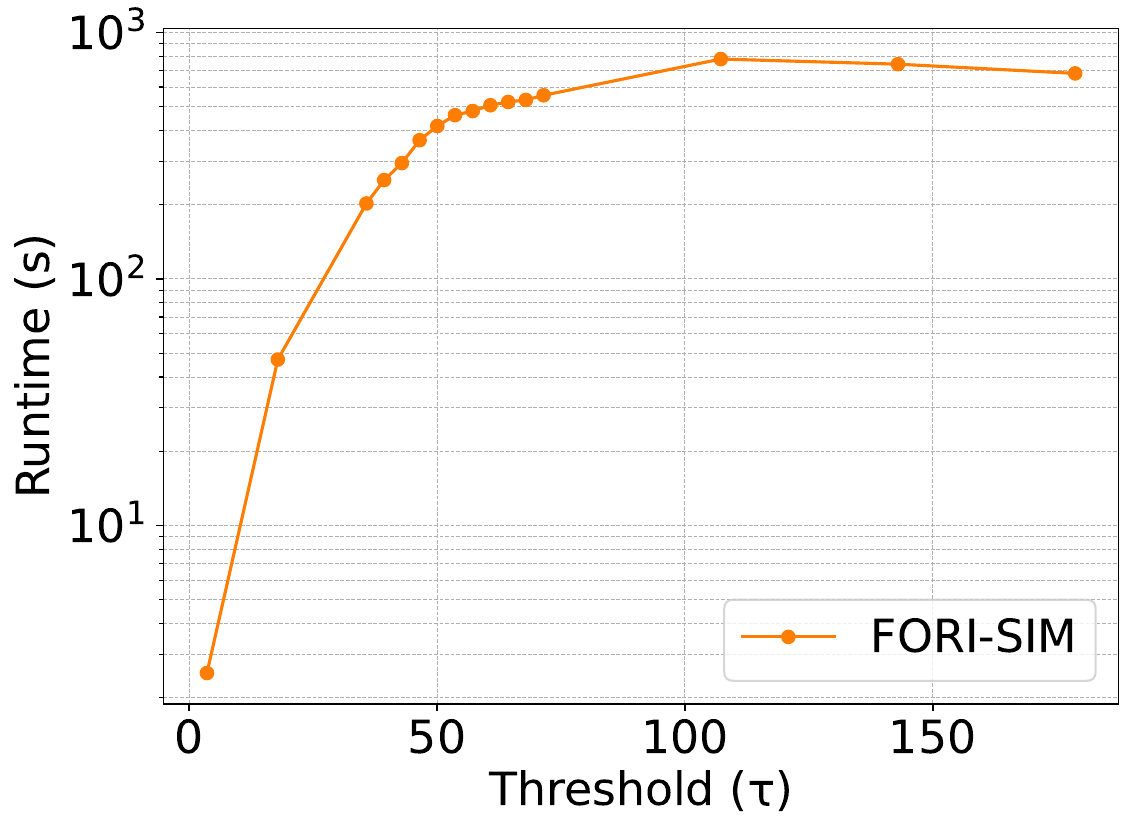}
        \caption{\aids}
        \label{fig:lineplot_aids_nonuni}
  \end{subfigure}
  \\
  \begin{subfigure}[b]{.8\linewidth}
      \centering
        \includegraphics[width=\linewidth]{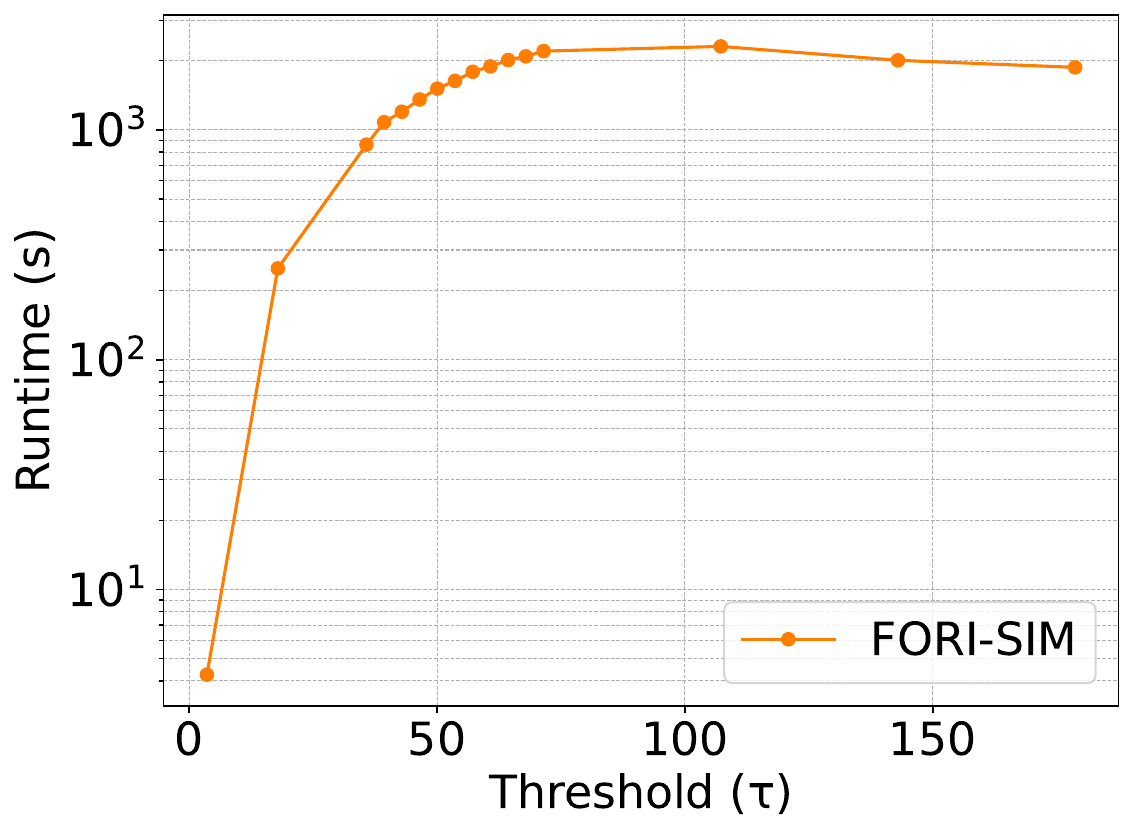}
        \caption{\muta}
        \label{fig:lineplot_muta_nonuni}
  \end{subfigure}
  \\
  \begin{subfigure}[b]{.8\linewidth}
      \centering
        \includegraphics[width=\linewidth]{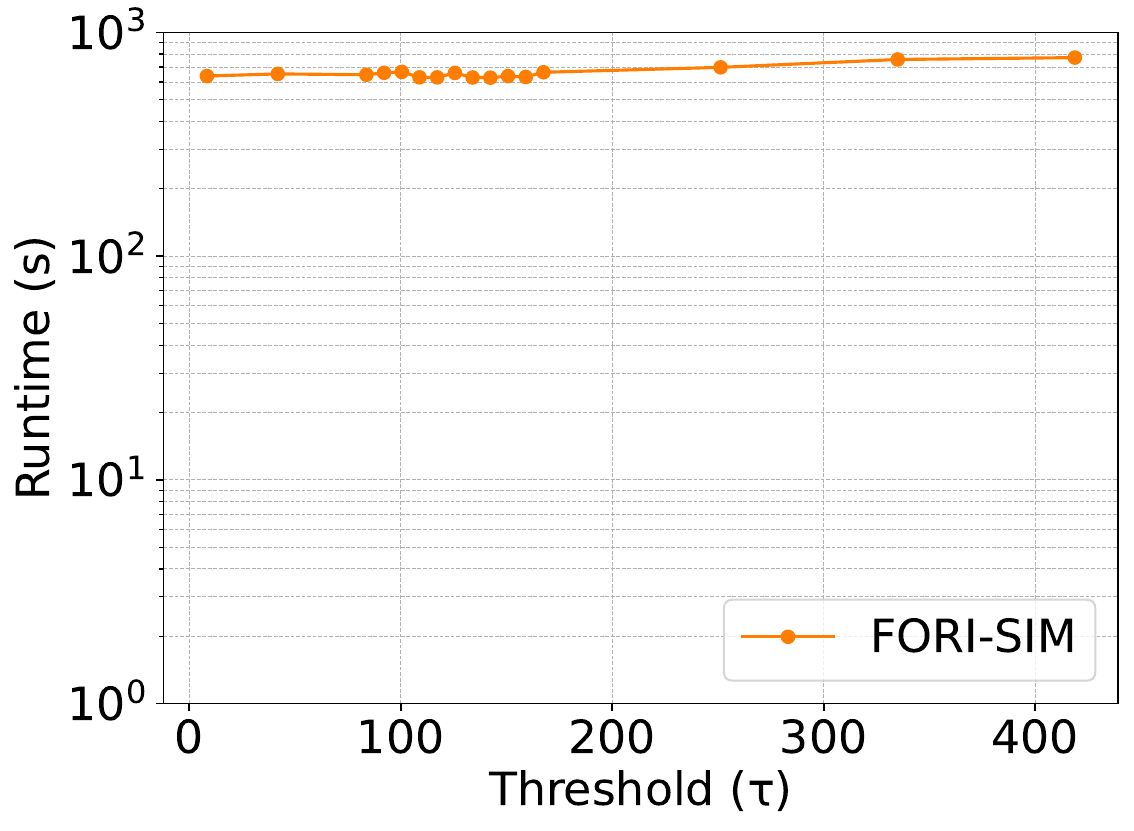}
        \caption{\protein}
        \label{fig:lineplot_prot_nonuni}
  \end{subfigure}
  \caption{Line plots of aggregated runtime performance of \forisim on the graph similarity search problem on non-uniform costs as a function of $\tau$.}
  \label{fig:lineplots_nonuni}
\end{figure}

\section{Conclusions}
In this paper, we propose \forisim, an algorithm to solve the graph similarity search under both uniform and non-uniform edit cost functions. This is the first to use integer programming techniques and employs a hierarchy of lower bounds including our novel lower bound \forilp. We theoretically establish its improvement over state-of-the-art lower bounds. Extensive experimental evaluation underlines the impressive quality of the lower bound, which is achieved with a reasonable runtime trade off. Furthermore, we show that our novel graph similarity search algorithm \forisim outperforms the state-of-the-art algorithm \bmaoged on all except the smallest of thresholds both in terms of runtime and memory consumption highlighting its superior scalability.

Possibilities for future work include comparing \forilp and \bm when considering partial mappings, where a subset of vertices is already fixed. Further, it is possible to incorporate algorithms that provide upper bounds into \forisim to possibly accept graphs earlier, or devise tighter and more efficient lower bounds. 

\bibliographystyle{IEEEtran}
\bibliography{merged_bibliography}
\vspace{12pt}

\end{document}